\documentclass{lmcs} %%% last changed 2014-08-20
\pdfoutput=1

\usepackage{lastpage}
\lmcsdoi{15}{2}{4}
\lmcsheading{}{\pageref{LastPage}}{}{}%
{Jan.~17,~2018}{Apr.~26,~2019}{}


\usepackage{enumitem,hyperref}
\theoremstyle{plain} %\crefname{satz}{Satz}{S\"atze}
\usepackage{microtype}%if unwanted, comment out or use option "draft"

\usepackage[xspace]{ellipsis}
\usepackage{algorithm2e}
\usepackage{listings}
\usepackage{tikz}
\usepackage{proof,amssymb}

\usepackage{scalerel}
\usepackage{csquotes}
\usepackage{booktabs} 

\usepackage{amsmath}

\usepackage{thmtools,thm-restate}
\usepackage{centernot}
\usepackage{mathtools}
\usepackage{stmaryrd}

  \lstdefinelanguage{pseudo}{
    morekeywords={if,elseif,then,return,end,choose,guess,when,for,foreach,case},
    morekeywords=[3]{false,true,and,or,not},
    morecomment=[l]{//}
  }
\newcommand{\ie}{i.\,e.\@\xspace}

\newcommand{\Wloss}{W.l.o.g.\@\xspace}

\newcommand{\complClFont}[1]{\mathbf{#1}}         % P, NP, NL, ...
\newcommand{\logicClFont}[1]{\mathsf{#1}}        % LTL, CTL, ...
\newcommand{\problemFont}[1]{\mathrm{#1}}         % SAT, TAUT, IMP, ...
\newcommand{\mathCommandFont}[1]{\mathrm{#1}}     % max, min, ...

\newcommand{\bigO}[1]{\protect\ensuremath{{\mathcal{O}\left(#1\right)}}} % does not grow stronger
\newcommand{\bin}[1]{{\protect\ensuremath{\mathCommandFont{bin}(#1)}}}

\newcommand{\tuple}[1]{\vec{#1}}
\newcommand{\sub}{\subseteq}

\newcommand{\Fr}[1]{{\protect\ensuremath{\mathsf{Fr}(#1)}}}
\newcommand{\Var}[1]{{\protect\ensuremath{\mathsf{Var}(#1)}}}

\newcommand{\Props}{{\protect\ensuremath{\mathsf{Prop}}}}

\newcommand {\indep}[3] {#2 ~\bot_{#1}~ #3}

\providecommand{\dfn}{\mathrel{\mathop:}=}
\providecommand{\ddfn}{\mathrel{\mathop{{\mathop:}{\mathop:}}}=}

\newcommand{\ADQBF}{\protect\ensuremath\problemFont{ADQBF}}
\newcommand{\DQBF}{\protect\ensuremath\problemFont{DQBF}}

\newcommand{\TRUE}{\protect\ensuremath\problemFont{TRUE}}

\newcommand{\PSPACE}{\protect\ensuremath{\complClFont{PSPACE}}\xspace}
\newcommand{\NEXPTIME}{\protect\ensuremath{\complClFont{NEXPTIME}}\xspace}
\newcommand{\coNEXPTIME}{\protect\ensuremath{\complClFont{co\textrm{-}NEXPTIME}}\xspace}
\newcommand{\EXPTIME}{\protect\ensuremath{\complClFont{EXPTIME}}\xspace}
\newcommand{\SigmaE}[1]{{\ensuremath\protect\Sigma^\mathCommandFont{EXP}_{#1}}}

\newcommand{\NP}{\protect\ensuremath{\complClFont{NP}}\xspace}
\newcommand{\PiE}[1]{{\ensuremath\protect\Pi^\mathCommandFont{EXP}_{#1}}}

\newcommand{\coNP}{\protect\ensuremath{\complClFont{co\textrm{-}NP}}\xspace}
\newcommand{\SigmaQBF}[1]{{\ensuremath\protect\Sigma_{#1}\textrm{-}\ADQBF}}
\newcommand{\PiQBF}[1]{{\ensuremath\protect\Pi_{#1}\textrm{-}\ADQBF}}

\newcommand{\calA}{\protect\ensuremath{\mathcal{A}}}
\newcommand{\calB}{\protect\ensuremath{\mathcal{B}}}
\newcommand{\calC}{\protect\ensuremath{\mathcal{C}}}
\newcommand{\calD}{\protect\ensuremath{\mathcal{D}}}

\newcommand{\calL}{\protect\ensuremath{\mathcal{L}}}
\newcommand{\calM}{\protect\ensuremath{\mathcal{M}}}

\newcommand{\calP}{\protect\ensuremath{\mathcal{P}}}

\newcommand{\PL}{\logicClFont{PL}}
\newcommand{\PLDep}{\logicClFont{PDL}}
\newcommand{\PLInc}{\logicClFont{PLInc}}
\newcommand{\PLInd}{\logicClFont{PLInd}}

\newcommand{\ML}{\logicClFont{ML}}
\newcommand{\MLDis}{\logicClFont{ML}(\cvee)}
\newcommand{\PLDis}{\logicClFont{PL}(\cvee)}
\newcommand{\QPLDis}{\logicClFont{QPL}(\cvee)}
\newcommand{\RML}{\logicClFont{RML}}
\newcommand{\MLDep}{\logicClFont{MDL}}
\newcommand{\EMLDep}{\logicClFont{EMDL}}
\newcommand{\MLInc}{\logicClFont{MLInc}}
\newcommand{\MLInd}{\logicClFont{MLInd}}

\newcommand{\QPL}{\logicClFont{QPL}}

\newcommand{\QPLDep}{\logicClFont{QPDL}}

\newcommand{\QPLInc}{\logicClFont{QPLInc}}
\newcommand{\QPLInd}{\logicClFont{QPLInd}}

\newcommand{\LL}{\logicClFont{L}}

\newcommand{\dep}[1]{\mathrm{dep}\!\left(#1\right)}

\makeatletter
\newcommand{\xMapsto}[2][]{\ext@arrow 0599{\Mapstofill@}{#1}{#2}}
\def\Mapstofill@{\arrowfill@{\Mapstochar\Relbar}\Relbar\Rightarrow}
\makeatother

  \newcommand{\Branch}[2]{\protect\ensuremath{\mathsf{branch}_{#2}(#1)}}

  \newcommand{\Store}[2]{\protect\ensuremath{\mathsf{store}_{#2}(#1)}}
  
 \newcommand{\tree}[1]{\protect\ensuremath{\mathsf{tree}(#1)}}
 \newcommand{\cveee}{\hspace{.25mm}\varovee\hspace{.25mm}}

\newcommand{\cvee}{\scalebox{1.3}{\mbox{$\cveee$}}}
\newcommand{\bcvee}{\scalebox{2}{\mbox{$\cveee$}}}
\renewcommand{\vec}{\overline}
\newcommand{\bigcvee}[1]{\underset{#1}{\bcvee}}
\begin{document}

\title[Validity and Entailment in Modal and Propositional Dependence Logics]{Validity and Entailment in Modal and Propositional Dependence Logics}
%\titlecomment{{\lsuper*}OPTIONAL comment concerning the title, \eg, 
 % if a variant or an extended abstract of the paper has appeared elsewhere.}

\author[M. Hannula]{Miika Hannula}	%required
\address{Department of Computer Science, University of Auckland, New Zealand}
\address{Department of Mathematics and Statistics, University of Helsinki, Finland}	%required
\email{miika.hannula@helsinki.fi}  %optional
\thanks{This work was supported by the Marsden Fund grant UOA1628 and the Academy of Finland grant 308712.}	%optional

\begin{abstract}
%The computational properties of modal and propositional dependence logics have been extensively studied over the past few years, starting from a result by Sevenster showing $\NEXPTIME$-completeness of the satisfiability problem for modal dependence logic. Thus far, however, the validity and entailment properties of these logics have remained uncharacterised to a great extent. This paper establishes a complete classification of the complexity of validity and entailment in modal and propositional dependence logics. In particular, we address the question of the complexity of validity in modal dependence logic. By showing that it is $\NEXPTIME$-complete we refute an earlier conjecture proposing a higher complexity for the problem.
The computational properties of modal and propositional dependence logics have been extensively studied over the past few years, starting from a result by Sevenster showing $\NEXPTIME$-completeness of the  satisfiability problem for  modal dependence logic. Thus far, however, the validity and entailment properties of these logics have remained mostly unaddressed. This paper provides a comprehensive classification of the complexity of validity and entailment in various modal and propositional dependence logics. The logics examined are obtained by extending the standard modal and propositional logics with notions of dependence, independence, and inclusion in the team semantics context. 
In particular, we address the question of the complexity of validity in modal dependence logic. By showing that it is $\NEXPTIME$-complete we refute an earlier conjecture proposing a higher complexity for the problem.
\end{abstract}

\maketitle

\section{Introduction}
The notions of dependence and independence are pervasive in various fields of science. 
%In physics the time of descent of an object depends only on the height and is independent of the weight; in genetics the sex of humans is determined by the XY-chromosomes; and in social choice theory the societal outcome in any ranked voting system satisfying unanimity and independence of irrelevant alternatives is determined by some single voter. 
Usually these concepts manifest themselves in the presence of \emph{multitudes} (e.g. events or experiments). 
Dependence logic \cite{vaananen07} is a recent logical formalism which, in contrast to others, has exactly these multitudes as its underlying concept. In this article we study dependence logic in the propositional and modal logic context and present a complete classification of the computational complexity of their associated entailment and validity problems.

%Many formal languages (e.g. first-order, propositional, and modal logic) 
In first-order logic, the standard formal language behind all mathematics and computer science,  dependencies between variables arise strictly from the order of their quantification. Consequently, more subtle forms of dependencies cannot be captured, a phenomenon exemplified by the fact that first-order logic lacks expressions for statements of the form 
\[\text{``for all $x$ there is $y$, and for all $u$ there is $v$,  such that $R(x,y,u,v)$''}\]
where  $y$ and $v$ are to be chosen independently from one another. To overcome this barrier, branching quantifiers of Henkin \cite{henkin61} and independence-friendly logic of Hintikka and Sandu \cite{hintikkasandu89} suggested the use of quantifier manipulation. Dependence logic instead extends first-order logic at the atomic level with the introduction of new dependence atoms 
\begin{equation}\label{depatomi}
\dep{x_1, \ldots  ,x_n}
\end{equation}
 which indicate that the value of $x_n$ depends only on the values of $x_1, \ldots ,x_{n-1}$. Dependence atoms are evaluated over \emph{teams}, i.e., sets of assignments %(or sets of Boolean assignments, worlds of a Kripke model, rows of a database, etc.)
 which form the basis of \emph{team semantics}. The concept of team semantics was originally proposed by Hodges in refutation of the view of Hintikka that the logics of imperfect information, such as his independence-friendly logic, escape natural compositional semantics \cite{hodges97}. By the development of dependence logic it soon became evident that  team semantics serves also as a connecting link between the aforenementioned logics and the relational database theory. In particular, team semantics enables the extensions of even weaker logics, such as modal and propositional logics, with various sophisticated dependency notions known from the database literature \cite{galliani12,gradel10,KontinenMSV14,KrebsMV15}. In this article we consider modal and propositional dependence logics that extend modal and propositional logics with dependence atoms similar to \eqref{depatomi}, the only exception being that dependence atoms here declare dependencies between propositions. We establish a complete classification of the computational complexity of the associated entailment and validity problems, including a solution to an open problem regarding the complexity of validity in modal dependence logic.

Modal dependence logic was introduced by  V{\"a}{\"a}n{\"a}nen in 2008 \cite{vaananen08b}, and soon after it was shown to enjoy a $\NEXPTIME$-complete satisfiability problem \cite{sevenster09b}.  Since then the 
 expressivity, complexity, and axiomatizability properties of modal dependence logic and its variants have been exhaustively studied. Especially the complexity of satisfiability and model checking for modal dependence logic and its variants has been already comprehensively classified \cite{ebbing13,ebbing12,HKVV15,HannulaKVV18,HellaKMV15,HellaS15,KontinenMSV14,KrebsMV15,LohmannV13,muller13}.  It is worth noting here that satisfiability and validity are not dual to each other in the dependence logic context. Dependence logic cannot express classical negation nor logical implication which renders its associated validity, satisfiability, and entailment problems genuinely different. Entailment and validity of modal and propositional dependence logics have
   been axiomatically characterized by Yang and V{\"a}{\"a}n{\"a}nen in \cite{yangarxiv16,yang14,Yang2016557} and also by Sano and Virtema in \cite{sano14}.  Nevertheless, the related complexity issues have remained almost totally unaddressed.  
The aim of this article is to address this shortage in research by presenting a comprehensive classification with regards to these questions. %It is all the more surprising as 

 A starting point for our endeavour is a recent result by Virtema which showed that the validity problem for propositional dependence logic is $\NEXPTIME$-complete \cite{Virtema14}. In that paper
 the complexity of validity for modal dependence logic remained unsettled, although it was conjectured to be harder than that for propositional dependence logic. This conjecture is refuted in this paper as the same exact $\NEXPTIME$ bound is shown to apply to modal dependence logic as well. Furthermore, we show that this result applies to  the  extension of propositional dependence  logic with quantifiers as well as to the so-called extended modal logic which can express dependencies between arbitrary modal formulae (instead of simple propositions). 
%Moreover, this result is shown to apply to the so-called extended modal dependence logic that extends the scope of dependence atoms to arbitrary modal formulae \cite{ebbing13}. 
These complexity bounds follow as corollaries from a more general result showing that the entailment problem for (extended) modal dependence and propositional dependence logics is complete for $\coNEXPTIME^{\NP}$.  We also consider modal logic extended with so-called intuitionistic disjunction and show that the associated entailment, validity, and satisfiability problems are all $\PSPACE$-complete, which is, in all the three categories the complexity of the standard modal logic. 

The aforementioned results have interesting consequences. First, combining results from this paper and \cite{ sevenster09b,Virtema14} we observe that similarly to the standard modal logic case the complexity of validity and satisfiability coincide for (extended) modal dependence logic. Secondly, it was previously known that propositional and modal dependence logics deviate on the complexity of their satisfiability problem ($\NP$-complete vs. $\NEXPTIME$-complete \cite{LohmannV13,sevenster09b}, resp.) and that the standard propositional and modal logics differ from one another on both satisfiability and validity ($\NP$-complete/$\coNP$-complete vs. $\PSPACE$-complete, resp. \cite{Cook71,Ladner77,Levin73}). 
Based on this it is somewhat surprising to find out that  modal and propositional dependence logics correspond to one another in terms of the complexity of both their validity and entailment problems.

%It follows that 
%(extended) modal dependence and propositional dependence logic correspond to one another in terms of the complexity of both their validity and entailment problems. However, satisfiability of propositional dependence logic is known to have a lower complexity than that of its modal counterpart ($\NP$-complete vs. $\NEXPTIME$-complete, resp. \cite{LohmannV13,sevenster09b,Virtema14}). In contrast, recall that the standard propositional and modal logic differ from another in all of the aforenementioned categories ($\NP$-complete/$\coNP$-complete vs. $\PSPACE$-complete, resp. \cite{Cook71,Ladner77,Levin73}).

We also establish exact complexity bounds for entailment and validity of quantified propositional independence and inclusion logics. These logics are extensions of propositional logic with quantifiers and either independence or inclusion atoms \cite{HLKV16}. 
%Quantified propositional logics can be seen as intermediate logics between propositional and modal logic. %Our findings show that this problem is $\coNEXPTIME$-hard for the first and $\coNEXPTIME^{\NP}$-hard for the second. 
We obtain our results by investigating recent generalizations of the quantified Boolean formula problem. % \cite{HLKV16}.  
 The validity and entailment problems for quantified propositional independence logic are both  shown to be 
 $\coNEXPTIME^{\NP}$-complete. For quantified propositional inclusion logic  entailment is shown to be $\coNEXPTIME$-complete whereas validity is only $\EXPTIME$-complete. 
 Using standard reduction methods the examined quantified propositional logics can be interpreted as fragments of modal independence and inclusion logics. %Therefore, the same lower bounds can be applied to 
 Our findings then imply that validity is harder for modal independence logic than it is for modal dependence logic (unless the exponential-time hierarchy collapses at a low level), although in terms of satisfiability both logics are $\NEXPTIME$-complete \cite{KontinenMSV14}. We refer the reader to Table \ref{newresults} for a summary of our results.

\noindent
\textbf{Organization.} This article is organized as follows. %In Section \ref{preli} we present  some notation and background assumptions. 
 In Section \ref{sect:mdl} we give a short introduction to modal dependence logics, followed by Section \ref{sect:upper} which proves $\coNEXPTIME^{\NP}$-membership for modal dependence logic entailment. In Section \ref{sect:prop} we define (quantified) propositional dependence logics, and in Section \ref{sect:lower} we show  $\coNEXPTIME^{\NP}$-hardness for entailment in this logic.   %present  basic properties of  modal and (quantified) propositional dependence logics, and give background for the proof techniques used in the paper. In Sections \ref{sect:mldep}, \ref{sect:mlind}, and \ref{sect:mlinc} we consider entailment in modal and propositional dependence, independence, and inclusion logic, respectively. %In Section \ref{sect:axioms} we introduce novel sound and complete axiomizations for (quantified) propositional dependence logic and extended modal dependence logic, based on our observations on their associated entailment problems. 
  In Section \ref{sect:mldep} these findings are drawn together to establish exact complexity bounds. In Sections \ref{sect:mlind} and \ref{sect:mlinc} we shift focus to validity and entailment of modal and quantified propostitional logics defined in terms of independence and inclusion atoms, respectively.
Finally, Section \ref{sect:conclusion} is  reserved for conclusions. %For some of the proofs we refer the reader to Appendix.
%\noindent
%\textbf{Remarks} %Following the common convention in the team semantics context, all our logics in this section and later  will  contain only formulae in negation normal form (NNF). 
 Throughout the paper we assume that the reader is familiar with the basic concepts of propositional and modal logic, as well as those of computational complexity.  All the hardness results in the paper are stated under polynomial-time reductions. %On general notational conventions,

\section{Modal Dependence Logics}\label{sect:mdl}
In this section we introduce extensions of modal logic with dependencies and present some of their basic properties. Following the common convention  in team semantics we restrict attention to formulae in negation normal form (NNF). Note that the concept of negation referred to here is not classical. In the team semantics context \enquote{$\neg$} is used to express that something is not true for all individual members of a team. Later in this section we also present an auxiliary relational variant of modal logic for pointed Kripke models to facilitate the upper bound proof of the next section. For this variant we employ negation classically and do not restrict its scope in formulae. %whereas classical negation (usually denoted in team semantics by \enquote{$\sim$})  would express that something is not true for a team.

The syntax of \emph{modal logic} ($\ML$) is generated by the following grammar:
\begin{equation}\label{def:ml}
\phi\ddfn p \mid \neg p \mid (\phi \wedge \phi) \mid (\phi \vee \phi) \mid \Box \phi \mid \Diamond \phi.
\end{equation}
%\end{defi}
Extensions of modal logic with different dependency notions are made possible via a generalization of the standard Kripke semantics by teams. %, here defined  as sets of worlds.
  A \emph{Kripke model} over a set of variables $V$ is a tuple $\calM=(W,R,\pi)$ where $W$ is a non-empty set of worlds, $R$ is a binary relation over $W$, and $\pi\colon V\to \calP(W)$ is a function that associates each variable with a set of worlds. A \emph{team} $T$ of a Kripke model $\calM=(W,R,\pi)$  is a  subset of $ W$.  For the team semantics of modal operators we define  the set of successors of a team $T$ as $R[T]:=\{w\in W\mid \exists w'\in T: (w',w)\in R\}$ and the set of  successor teams of a team $T$ as $R\langle T\rangle :=\{T'\sub R[T]\mid \forall w\in T\hspace{1mm} \exists w'\in T': (w,w')\in R\}$. The team semantics of modal logic is  now defined as follows.
\begin{defi}[Team Semantics of $\ML$]\label{ts:ml}
Let $\phi$ be an $\ML$ formula, let $\calM=(W,R,\pi)$ be a Kripke model over $V\supseteq \Var{\phi}$, and let $T\sub W$. The satisfaction relation $\calM,T\models \phi$ is defined as follows:
\begin{align*}
\calM,T\models p \quad :\Leftrightarrow \quad &T\sub \pi(p),\\
\calM,T\models \neg p \quad :\Leftrightarrow \quad & T\cap \pi(p) = \emptyset,\\
\calM,T\models \phi_1\wedge \phi_2\quad :\Leftrightarrow \quad & \calM,T\models \phi_1 \textrm{ and }\calM,T\models \phi_2,\\
\calM,T\models \phi_1\vee \phi_2\quad :\Leftrightarrow \quad &\exists T_1,T_2:T_1\cup T_2=T, \calM,T_1\models \phi_1,\textrm{ and }\calM,T_2\models \phi_2,\\
\calM,T\models \Diamond \phi \quad :\Leftrightarrow \quad &\exists T'\in R\langle T\rangle: \calM,T'\models \phi,\\
\calM,T\models \Box \phi \quad :\Leftrightarrow \quad & \calM,R[T]\models \phi.
\end{align*}
\end{defi} 
We write $\phi\equiv \psi$ to denote that $\phi$ and $\psi$ are \emph{equivalent}, i.e., for all Kripke models $\calM$ and teams $T$, $\calM,T\models \phi $ iff $ \calM,T\models \phi'$.
Let  $\Sigma\cup\{\phi\}$ be a set of formulae. We write $\calM,T\models \Sigma$ iff $\calM,T\models \phi$ for all $\phi \in \Sigma$, and say that  $\Sigma$ \emph{entails} $\phi$ if for all $\calM$ and $T$, $\calM ,T\models \Sigma$ implies $\calM,T\models \phi$. Let $\calL$ be a logic in the team semantics setting. 
The \emph{entailment problem} for $\calL$ is to decide whether $\Sigma$ entails $\phi$ (written $\Sigma \models \phi$) for a given finite set of formulae $\Sigma\cup\{\phi\}$ from $\calL$. The \emph{validity problem} for $\calL$ is to decide whether a given formula $\phi\in \calL$ is satisfied by all Kripke models and teams. The \emph{satisfiability problem} for $\calL$ is to decide whether a given formula $\phi\in \calL$ is satisfied by some Kripke model and a non-empty team\footnote{The empty team satisfies all formulae trivially.}.
 %By replacing $\models$ with or $\models_{\ML}$ we refer to logical consequence in terms of the standard modal logic semantics, respectively. 
 
The following flatness property holds for all modal logic formulae. Notice that by $\models_{\ML}$ we refer to the usual satisfaction relation of modal logic.
\begin{prop}[Flatness \cite{sevenster09b}]\label{prop:ml_flatness}
Let $\phi$ be a formula in $\ML$, let $\calM=(W,R,\pi)$ be a Kripke model over $V\supseteq \Var{\phi}$, and let $T\sub W$ be a team. Then:
\begin{align*}
\calM,T\models \phi \quad \Leftrightarrow \quad & \forall w\in T: \calM,w\models_{\ML} \phi.
\end{align*}
\end{prop}

Team semantics gives rise to different extensions of modal logic capable of expressing various dependency notions. In this article we consider dependence atoms that express functional dependence between propositions. To facilitate their associated semantic definitions, we first define for each world $w$ of a Kripke model $\calM$ a truth function $w_{\calM}$ from $\ML$ formulae into $\{0,1\}$ as follows:
\[w_{\calM}(\phi) = \begin{cases} 1 &\text{ if }\calM ,\{w\} \models \phi,\\
0 & \text{ otherwise.}
\end{cases}
\]
\textbf{1) Dependence.} \emph{Modal dependence logic} ($\MLDep$) is obtained by extending $\ML$ with \emph{dependence atoms}
\begin{equation}\label{depatom}
\dep{\tuple p, q}
\end{equation}
where $\tuple p$ is a sequence of propositional variables and $q$ is a single propositional variable. 
Furthermore, we consider \emph{extended dependence atoms} of the form 
\begin{equation}\label{edepatom}
\dep{\tuple \phi, \psi}
\end{equation}
where $\tuple \phi$ is a sequence of  $\ML$ formulae and $\psi$ is a single $\ML$ formula. The extension of $\ML$ with formulae of the form \eqref{edepatom} is called \emph{extended modal dependence logic} ($\EMLDep$). Formulae of the form \eqref{depatom} and \eqref{edepatom} indicate that the (truth) value of the formula on the right-hand side is functionally determined by the (truth) values of the formulae listed on the left-hand side. The satisfaction relation for both \eqref{depatom} and \eqref{edepatom} is defined accordingly as follows:
\begin{align*}
\calM ,T \models \dep{\tuple \phi, \psi} \quad :\Leftrightarrow \quad&\forall w,w'\in T: w_{\calM}(\tuple \phi)=w'_{\calM} (\tuple \phi) \text{ implies }w_{\calM}(\psi)=w'_{\calM}(\psi).
\end{align*}
Here and below we use $w_{\calM}(\tuple \theta)$ as a shorthand for $(w_{\calM}( \theta_1), \ldots ,w_{\calM}( \theta_n))$  if  $\tuple \theta$ is a sequence of formulae $(\theta_1, \ldots ,\theta_n)$.

We also examine so-called \emph{intuituionistic disjunction} $\cvee$ defined as follows:
\begin{align}\label{def:cvee}
\calM,T\models \phi_1\cvee\phi_2 \quad :\Leftrightarrow  \quad& \calM, T\models \phi_1 \textrm{ or }\calM,T\models \phi_2.
\end{align}
We denote the extension of $\ML$ with intuitionistic disjunction $\cvee$ by $\MLDis$. Notice that the logics $\MLDep$ and $\EMLDep$ are expressively equivalent to $\MLDis$ but exponentially more succinct as the translation of \eqref{depatom} to $\MLDis$ involves a necessarily exponential blow-up \cite{HellaLSV14}. These logics %$\MLDep$, $\EMLDep$, and  $\MLDis$ 
 satisfy the following downward closure property which will be used in the upper bound result.

\begin{prop}[Downward Closure \cite{ebbing13,vaananen08b,yang14}]\label{prop:ml_dc}
Let $\phi$ be a formula in $\MLDep$, $\EMLDep$, or $\MLDis$, let $\calM=(W,R,\pi)$ be a Kripke model over $V\supseteq \Var{\phi}$, and let $T\sub W$ be a team. Then:
\begin{align*}
T' \sub T\textrm{ and }\calM,T\models \phi \quad\Rightarrow \quad& \calM,T'\models \phi.
\end{align*}
\end{prop}

\noindent
\textbf{2) Independence.} \emph{Modal independence logic} $(\MLInd)$ extends $\ML$ with \emph{independence atoms} 
\begin{equation}\label{indatom}
\indep{\tuple p}{\tuple q}{\tuple r}
\end{equation}
where $\tuple p,\tuple q,\tuple r$ are sequences of propositional variables. Intuitively, \eqref{indatom} expresses that the  values of $\tuple q$ and $\tuple r$ are independent of one another, given any value of $\tuple r$.  The associated satisfaction relation is defined  as follows:
\begin{align*}\calM ,T \models \indep{\tuple p}{\tuple q}{\tuple r}\quad :\Leftrightarrow \quad&
\forall w,w'\in T: w_{\calM}(\tuple p)=w'_{\calM} (\tuple p) 
\text{ implies }\\
&\exists w''\in T: w_{\calM}(\tuple p\tuple q)=w''_{\calM} (\tuple p\tuple q) 
\text{ and }\\
&w'_{\calM}(\tuple r)=w''_{\calM} (\tuple r) .
\end{align*}
 The definition expresses that, fixing any values for $\tuple p$, the  values for  $\tuple q \tuple r$ form a cartesian product defined in terms of the values for $\tuple q$ and $\tuple r$. Furthermore, notice that $\MLInd$ subsumes $\MLDep$ since \eqref{depatom}  can be expressed by $\indep{\tuple p}{q}{q}$.

\noindent
\textbf{3) Inclusion.} \emph{Modal inclusion logic} ($\MLInc$) extends $\ML$ with \emph{inclusion atoms} 
\begin{equation}\label{incatom}
\tuple p\sub \tuple q
\end{equation}
where $\tuple p$ and $\tuple q$ are sequences of propositional variables of the same length. This atom indicates that the values of $\tuple q$ subsume all the values of $\tuple p$. The satisfaction relation for \eqref{incatom} is defined as follows:
\[\calM ,T \models p\sub \tuple q\quad :\Leftrightarrow \quad \forall w\in T \exists w': w_{\calM}(\tuple p)=w'_{\calM} (\tuple q) .\]

For the sake of our proof arguments, we also extend modal logic with predicates. The syntax of \emph{relational modal logic} ($\RML$) is given by the grammar:
\begin{equation}\label{def:rml}
\phi\ddfn p  \mid \neg \phi \mid  (\phi \wedge \phi)  \mid \Box \phi \mid S(\phi_1, \ldots ,\phi_n).
\end{equation}
 %  is obtained by extending  $\ML$  with relational and negated relational atoms $S(\phi_1, \ldots ,\phi_n)$ and $\neg S(\phi_1, \ldots ,\phi_n)$ where $\phi_1, \ldots ,\phi_n$ are $\ML$ formulae. 
 The formulae of $\RML$ are evaluated over a world $w\in W$ in a \emph{relational Kripke model} $\calM=(W,R,\pi, S^{\calM}_1, \ldots ,S^{\calM}_n)$ where each $S^{\calM}_i$ is a set of binary sequences of length $\#S_i$, that is, the arity of the relation symbol $S_i$. We denote by $\calM,w\models_{\RML} \phi$ the satisfaction relation obtained by extending the standard Kripke semantics of modal logic as follows:
\begin{align*}
\calM,w\models_{\RML} S(\phi_1, \ldots ,\phi_n)  :\Leftrightarrow  
 (w_{\calM}(\phi_1), \ldots ,w_{\calM}(\phi_n))\in S^{\calM}.
\end{align*}
In particular, this means that negation is treated classically:
\begin{equation*}\calM,w\models_{\RML} \neg\phi:\Leftrightarrow \calM,w\not\models_{\RML} \phi.
\end{equation*}
 We also employ the usual shorthands $\phi \vee \psi:= \neg (\neg \phi \wedge \neg \psi)$ and $\Diamond \phi:= \neg \Box \neg \phi$ in $\RML$. %Note that \enquote{$\negg$} is distinguished from \enquote{$\neg$} which is used only in the team context in front of atomic formulae to indicate a specific type of negation over teams. Apart from $\RML$ which operates only in pointed Kripke models, we sometimes employ \enquote{$\sim$} in the team context analogously to \eqref{negg}, e.g., $\calM\models_T \negg \phi :\Leftrightarrow \calM \not\models_T \phi$ for modal formulae $\phi$.
 
We can now proceed to the upper bound result which states that the entailment problem for $\EMLDep$ is  in $\coNEXPTIME^{\NP}$.

\section{Upper Bound for $\EMLDep$ Entailment}\label{sect:upper}
In this section we  show that $\EMLDep$ entailment is in $\coNEXPTIME^{\NP}$.
  The idea is to represent dependence atoms using witnessing functions  guessed universally   on the left-hand side and existentially on the right-hand side of an entailment problem $\{\phi_1, \ldots ,\phi_{n-1}\} \models \phi_n$. This reduces the problem to  validity  of an $\RML$ formula of the form $\neg \phi^*_1\vee \ldots \vee \neg \phi^*_{n-1} \vee \phi_n^*$ where $\phi^*_i$ is obtained by replacing in $\phi_i$ all dependence atoms with relational atoms whose interpretations are bound the guessed Boolean functions. We then extend an Algorithm by Ladner that shows a $\PSPACE$ upper bound for the validity problem of modal logic \cite{Ladner77}. As a novel algorithmic feature  we introduce recursive steps  for relational atoms that query to the guessed functions. The $\coNEXPTIME^{\NP}$ upper bound then follows by a straightforward running time analysis.
  
We start by showing how to represent dependence atoms using intuitionistic disjunctions  defined over witnessing functions. We use $\phi^{\bot}$  to denote the NNF formula obtained from $\neg \phi$ by pushing the negation to the atomic level, and $\phi^{\top}$ to denote $\phi$. Let $\tuple \alpha=(\alpha_1,  \ldots ,\alpha_n)$ be a sequence of $\ML$ formulae and let $\beta$ be a single $\ML$ formula. Then we say that a function $f:\{\top,\bot\}^{n}\to\{\top,\bot\}$ is a \emph{witness} of $d:= \dep{\tuple \alpha,\beta}$, giving rise to a witnessing $\ML$ formula 
\begin{equation}\label{eqdep}
D(f,d):=\bigvee_{a_1, \ldots ,a_n\in \{\top,\bot\}} \alpha_1^{a_1} \wedge \ldots \alpha_n^{a_n} \wedge \beta^{f(a_1, \ldots ,a_n)}.
\end{equation}
The equivalence
\begin{equation}\label{eqdep2}
d \equiv \bigcvee{f\colon \{\top,\bot\}^{n}\to\{\top,\bot\}}  D(f,d)
\end{equation}
has been noticed in the contexts of $\MLDep$ and $\EMLDep$ respectively in \cite{vaananen08b,ebbing13}. %Note that a  
 %representation of the sort \eqref{eqdep2} necessitates that the represented formula, in this case the dependence atom $d$, has the downward closure property. 
 
%Our proof uses this observation and relates to the standard $\PSPACE$ algorithm for $\ML$ satisfiability \cite{Ladner77}. However, we wish 
To avoid the exponential blow-up involved in both  \eqref{eqdep} and \eqref{eqdep2}, we instead relate to $\RML$ by  utilizing the following equivalence:
 \begin{equation}\label{eqsuc}
 (W,R,\pi),w\models_{\ML} D(f,d) \Leftrightarrow (W,R,\pi, S^{\calM}),w\models_{\RML} S( \tuple \alpha\beta),
 \end{equation}
 where $S^{\calM}:=\{(a_1, \ldots ,a_n,b)\in \{0,1\}^{n+1}\mid f(a_1, \ldots ,a_n)=b\}$.
Before proceeding to the proof, we need the following simple proposition, based on \cite{Virtema14,yang14} where the statement has been proven for empty $\Sigma$. 
\begin{prop}\label{prop:yang}
Let $\Sigma$ be a set of $\ML$ %$\calL$ formulae where $\calL \in \{\PL,\QPL,\ML\}$,
formulae, and let $\phi_0,\phi_1 \in 
\ML(\cvee)$. Then $\Sigma \models \phi_0\cvee\phi_1$ iff $\Sigma \models \phi_0$ or $\Sigma \models \phi_1$.
  \end{prop}
\begin{proof}
  It suffices to show the only-if direction. %Assume first that $\calL= \ML$, and l
  Let $\calM_0,T_0$ and $\calM_1,T_1$ be counterexamples to $\Sigma \models \phi_0$ and $\Sigma \models \phi_1$, respectively. \Wloss we may assume that $\calM_0$ and $\calM_1$ are disjoint. Since the truth value of an  $\MLDis$ formula is preserved under taking disjoint unions of Kripke models (see Theorem 6.1.9 in \cite{yang14}, also Corollary 5.7 in \cite{Virtema14}) we note that $\calM,T_0$ %\models \Sigma\cup\{\negg \phi_0\}$ 
   and  $\calM,T_1$ are also counterexamples to $\Sigma \models \phi_0$ and $\Sigma \models \phi_1$, respectively. Let $T:=T_0\cup T_1$. %\models \Sigma\cup\{\negg \phi_1\}$ where  $\calM=\calM_0\cup\calM_1$. 
   By the flatness property of $\ML$ (Proposition \ref{prop:ml_flatness}) $\calM,T\models \Sigma$, and by the downward closure property of $\MLDis$ (Proposition \ref{prop:ml_dc}) $\calM,T\not\models \phi_i$ for $i=0,1$. Consequently, $\Sigma \not\models \phi_0 \cvee \phi_1$ which concludes the proof.  %, and by the flatness property of $\ML$ (Proposition \ref{prop:ml_flatness}), we then obtain that $\calM,T\models \Sigma\cup\{\negg \phi_0, \negg\phi_1\}$ where 
    \end{proof}
  The proof now proceeds via Lemmata \ref{lem:apu} and \ref{lem:alg} of which the former constitutes the basis for our alternating exponential-time algorithm.  %, where $x$ is a team or a world.
Note that if $\phi$ is an $\EMLDep$ formula with $k$ dependence atom subformulae, listed (possibly with repetitions) in $d_1, \ldots d_k$, then we call $\tuple f=(f_1, \ldots ,f_k)$ a \emph{witness sequence} of $\phi$ if each $f_i$ is a witness  of $d_i$. Furthermore, we denote by $\phi(\tuple f /\tuple d)$ the $\ML$ formula obtained from $\phi$ by replacing each $d_i$ with $D(f_{i},d_i)$.

\begin{lem}\label{lem:apu}
Let $\phi_1,\ldots ,\phi_{n}$ be formulae in $\EMLDep$. Then $\{\phi_1, \ldots ,\phi_{n-1}\}\models \phi_n$ iff for all witness sequences $\tuple f_1, \ldots ,\tuple f_{n-1}$ of $\phi_1, \ldots ,\phi_{n-1}$ there is a witness sequence $\tuple f$ of $\phi_n$ such that
\[
 \{\phi_1(\tuple f_1/\tuple d_1), \ldots , \phi_{n-1}(\tuple f_{n-1}/\tuple d_{n-1})\}\models \phi_n(\tuple f_n/\tuple d_n).
\]
\end{lem}
\begin{proof}
Assume first that $\phi$ is an arbitrary formula in $\EMLDep$, and let $d=\dep{\tuple \alpha,\beta}$ be a subformula of $\phi$.  It is straightforward to show that $\phi$ is equivalent to
\[
\bigcvee{f\colon \{\top,\bot\}^{|\tuple \alpha|}\to\{\top,\bot\}} \phi(D(f,d)/d).
\]
This follows from the fact that all $\vee,\wedge,\Diamond,\Box$ distribute over $\cvee$, note especially that $(\phi\cvee \psi) \vee \theta$ is equivalent to $(\phi \vee\theta)\cvee (\psi \vee \theta)$.

Iterating these substitutions we  obtain that $\{\phi_1, \ldots,\phi_{n-1}\}\models \phi_n$ iff
\begin{equation}\label{eqxx}
\{\bigcvee{\tuple f_1} \phi_i(\tuple f_1/\tuple d_1), \ldots ,\bigcvee{\tuple f_{n-1}} \phi_i(\tuple f_{n-1}/\tuple d_{n-1})\}\models \bigcvee{\tuple f_n} \phi_i(\tuple f_n/\tuple d_n),
\end{equation}
where $\tuple f_i$ ranges over the witness sequences of $\phi_i$. Then \eqref{eqxx} holds iff for all $\tuple f_1, \ldots ,\tuple f_{n-1}$,
\begin{equation}\label{eqv}
\{  \phi_1(\tuple f_1/\tuple d_1), \ldots , \phi_{n-1}(\tuple f_{n-1}/\tuple d_{n-1})\} \models \bigcvee{\tuple f_n} \phi(\tuple f_n/\tuple d_n).
\end{equation}
Notice that each formula $\phi_i(\tuple f_i/\tuple d_i)$ belongs to $\ML$. Hence, by Proposition \ref{prop:yang} we conclude that \eqref{eqv} holds iff for all $\tuple f_1, \ldots ,\tuple f_{n-1}$ there is $\tuple f_n$ such that 
\begin{equation}\label{eqww}
 \{\phi_1(\tuple f_1/\tuple d_1), \ldots , \phi_{n-1}(\tuple f_{n-1}/\tuple d_{n-1})\}\models \phi(\tuple f_n/\tuple d_n).
\end{equation}
\end{proof}
The next proof step is to reduce an entailment problem of the form \eqref{eqww} to a validity problem of an $\RML$ formula over relational Kripke models whose interpretations agree with the guessed functions. For the latter problem we then apply Algorithm \ref{alg:mlsat} whose lines 1-14 and 19-26 constitute an algorithm of Ladner that shows the $\PSPACE$ upper bound for modal logic satisfiability \cite{Ladner77}. Lines 15-18 consider  those cases where the subformula is relational. Lemma \ref{lem:alg} now shows that, given an oracle $A$, this extended algorithm yields a $\PSPACE^A$ decision procedure for  satisfiability  of $\RML$ formulae over relational Kripke models whose predicates agree with $A$. For an oracle set $A$ of words from $\{0,1,\#\}^*$ and $k$-ary relation symbol $R_i$, we define $R_i^A:=\{(b_1, \ldots ,b_k)\in\{0,1\}^k\mid \bin{i}^\frown \#b_1\ldots b_k \in A\}$. Note that  $a^{\frown} b$ denotes the concatenation of two strings $a$ and $b$.

  \begin{figure}[h]
 \begin{algorithm}[H]\label{alg:mlsat}
 \caption{$\PSPACE^A$ algorithm for deciding validity in $\RML$. Notice that queries to $S_i^A$ range over $(b_1, \ldots ,b_k)\in \{0,1\}^k$.}
  \SetAlgoLined
  \LinesNumbered
  \SetKwInOut{Input}{Input}\SetKwInOut{Output}{Output}
  \SetKwFunction{KwFn}{Sat}
  \Input{$(\calA,\calB,\calC,\calD)$ where $\calA,\calB,\calC,\calD\sub \RML$}
  \Output{\KwFn{$\calA,\calB,\calC,\calD$}}
  \BlankLine		
  		\uIf{$\calA\cup\calB\not\sub\Props$}{
  			choose $\phi\in ((\calA\cup \calB)\setminus \Props)$\;
  			\uIf{$\phi =\neg \psi$ and $\phi \in \calA$}{
  				\Return \KwFn{$\calA\setminus\{\phi\},\calB\cup\{\psi\},\calC,\calD$}\;
  			}
  			\uElseIf{$\phi =\neg \psi$ and $\phi \in \calB$}{
  				\Return \KwFn{$\calA\cup\{\psi\},\calB\setminus\{\phi\},\calC,\calD$}\;
  			}
  			\uElseIf{$\phi = \psi\wedge \theta$ and $\phi\in \calA$}{
  				\Return \KwFn{$(\calA\cup\{\psi,\theta\})\setminus\{\phi\},\calB,\calC,\calD$}\;
			}
			\uElseIf{$\phi = \psi\wedge \theta$ and $\phi \in \calB$}{
				\Return \KwFn{$\calA,(\calB\cup\{\psi\})\setminus\{\phi\},\calC,\calD$}$\vee$ \KwFn{$\calA,							(\calB\cup\{\theta\})\setminus\{\phi\},\calC,\calD$}\;
			}
			\uElseIf{$\phi = \Box \psi$ and $\phi\in \calA$}{
				\Return \KwFn{$\calA \setminus\{\phi\},\calB,\calC\cup\{\psi\},\calD$}\;
			}
			\uElseIf{$\phi = \Box \psi$ and $\phi\in \calB$}{
				\Return \KwFn{$\calA ,\calB\setminus\{\phi\},\calC,\calD\cup\{\psi\}$}\;
			}
			\uElseIf{$\phi  =S_i(\psi_1, \ldots, \psi_k)$ and $\phi \in\calA$}{
					\Return $\bigvee_{(b_1, \ldots ,b_k)\in S^A_i}$\KwFn{$(\calA\cup\{\psi_{j}:b_j=1\})\setminus\{\phi\},													\calB\cup\{\psi_j:b_j=0\},\calC,\calD$}\;
			}
			\ElseIf{$\phi =S_i(\psi_1, \ldots ,\psi_k)$ and $\phi \in\calB$}{
					\Return $\bigvee_{(b_1, \ldots ,b_k)\not\in S^A_i}$\KwFn{$\calA\cup\{\psi_{j}:b_j=1\},							(\calB\cup\{\psi_j:b_j=0\})\setminus \{\phi\},\calC,\calD$}\;
  			}		
  		}		
  		\ElseIf{$(\calA\cup\calB)\sub\Props$}{
  			\uIf{$\calA\cap\calB\neq\emptyset$}{
  				\Return \textbf{false}\;
			}
			\uElseIf{$\calA\cap\calB= \emptyset$ and $\calC\cap\calD\neq\emptyset$}{
				\Return $\bigwedge_{D\in\calD}$\KwFn{$\calC,\{D\},\emptyset,\emptyset$}\;
			}
			\ElseIf{$\calA\cap\calB= \emptyset$ and $\calC\cap\calD=\emptyset$}{
				\Return \textbf{true}\;
			}
		}

\end{algorithm}
\end{figure}
%Algorithm \ref{alg:mlsat} extends an algorithm in \cite{Ladner77} with new clauses for relational formulae.
\begin{lem}\label{lem:alg}
Given an $\RML$ formula $\phi$ over a vocabulary $\{S_1, \ldots ,S_n\}$ and an oracle set of words $A$ from $\{0,1,\#\}^*$, Algorithm \ref{alg:mlsat} decides in $\PSPACE^A$ whether there is a relational Kripke structure $\calM=(W,R,\pi,S^A_1, \ldots ,S^A_n)$ and a world $w\in W$ such that $\calM,w\models_{\RML} \phi$.
\end{lem}
\begin{proof}
%We claim that Algorithm \ref{alg:mlsat} provides the wanted $\PSPACE^A$ decisision procedure. Notice that 
We leave it to the reader to show (by a straightforward structural induction) that, given an input $(\calA,\calB,\calC,\calD)$ where $\calA,\calB,\calC,\calD\sub \RML$, Algorithm \ref{alg:mlsat} returns $\texttt{Sat}(\calA,\calB,\calC,\calD)$ true iff there is a relational Kripke model $\calM=(W,R,\pi,S^{A}_1, \ldots ,S^{A}_n)$ and a world $w$ such that  
\begin{itemize}
\item $\calM,w\models_{\RML}\phi$ if  $\phi\in \calA$; 
\item $\calM,w\not\models_{\RML} \phi$ if  $\phi \in \calB$; 
\item $\calM,w\models_{\RML} \Box \phi$ if  $\phi \in \calC$; and 
\item $\calM,w\not\models_{\RML} \Box \phi$ if  $\phi \in \calD$. 
\end{itemize}
%Notice that lines 15-18 stipulate that selecting $\phi\in \calA$, where $\phi  =S_i(\psi_1, \ldots, \psi_k)$,  $\texttt{Sat}(\calA,\calB,\calC,\calD)$ returns 
%\[\bigvee_{(b_1, \ldots ,b_k)\in S^{A}_i}\texttt{Sat}(\calA\cup\{\psi_{j}:b_j=1\})\setminus\{\phi\},					\calB\cup\{\psi_j:b_j=0\},\calC,\calD),\]
%and selecting $\phi \in \calB$, where $\phi  =S_i(\psi_1, \ldots, \psi_k)$,  $\texttt{Sat}(\calA,\calB,\calC,\calD)$ returns
 %\[\bigvee_{(b_1, \ldots ,b_k)\in \{0,1\}^k\setminus S^{A}_i}\texttt{Sat}(\calA\cup\{\psi_{j}:b_j=1\},							(\calB\cup\{\psi_j:b_j=0\})\setminus \{\phi\},\calC,\calD).\]
Hence, $\texttt{Sat}(\{\psi\},\emptyset,\emptyset,\emptyset)$ returns true iff $\psi$ is satisfiable by  $\calM,w$ with relations $S^{\calM}_i$ obtained from the oracle. Note that the selection of subformulae $\phi$ from $\calA\cup \calB$ can be made deterministically by defining an ordering for the subformulae. Furthermore, we note, following \cite{Ladner77}, that this algorithm runs in $\PSPACE^{A}$ as it employs $\bigO{n}$ recursive steps that each take space $ \bigO{n}$. 
%Following \cite{Ladner77} we show that Algorithm \ref{alg:mlsat} requires only $\bigO{n^2}$ space on an input of the form $\texttt{Sat}(\{\phi\},\emptyset,\emptyset,\emptyset):$ it takes $\bigO{n}$ recursive steps, each taking space $ \bigO{n}$. 
\\% \\
\noindent
\textbf{Size of each recursive step.} At each recursive step  $\texttt{Sat}(\calA,\calB,\calC,\calD)$ is stored onto the work tape by listing all subformulae in $\calA\cup\calB\cup\calC\cup\calD$ in such a way that each subformula $\psi$ has its major connective (or relation/proposition symbol for atomic formulae) replaced with a special marker which also points to the position of the subset  where $\psi$ is located.  In addition we  store at each disjunctive/conjunctive recursive step the subformula or binary number that points to the disjunct/conjunct under consideration. Each recursive step takes now space $\bigO{n}$. \\
\noindent
\textbf{Number of recursive steps.}
Given a set of formulae $\calA$, we write $|\calA |$ for $\Sigma_{\phi\in\calA} |\phi |$ where $|\phi |$ is the length of $\phi$. We show by induction on $n =|\calA \cup\calB \cup\calC\cup\calD |$ that $\texttt{Sat}(\calA,\calB,\calC,\calD)$ has $2n+1$ levels of recursion. Assume that the claim holds for all natural numbers less than $n$, and assume that  $\texttt{Sat}(\calA,\calB,\calC,\calD)$ 
calls  $\texttt{Sat}(\calA',\calB',\calC',\calD')$. Then $|\calA' \cup\calB' \cup\calC'\cup\calD' |< n$ except for the case where $\calA\cap\calB$ is empty and $\calC\cap\calD$ is not. In that case it takes at most one extra recursive step to reduce to  a length $<n$. Hence, by the induction assumption the claim  follows. We conclude that the space requirement for Algorithm \ref{alg:mlsat} on $\texttt{Sat}(\{\phi\},\emptyset,\emptyset,\emptyset)$ is $\bigO{n^2}$. 
%A detailed space analysis (following that in \cite{Ladner77}) can be found in Appendix. 
% \begin{figure}[h]
 %\begin{algorithm}[H]   
% \LinesNumbered
 % \SetKwInOut{Input}{Input}\SetKwInOut{Output}{Output}
 % \SetKwFunction{KwFn}{Sat}
 %\Input{$(\calA,\calB,\calC,\calD)$ where $\calA,\calB,\calC,\calD\sub \RML$}
 % \Output{\KwFn{$\calA,\calB,\calC,\calD$}}
  %  \SetAlgoLined
   %\setcounter{AlgoLine}{14} 			
%			\uElseIf{$\phi  =S_i(\psi_1, \ldots, \psi_k)$, $\phi \in\calA$, and $\bin{i}=c_1\ldots c_m$}{
%					\Return $\bigvee_{(b_1, \ldots ,b_k)\in B_i}$\KwFn{$(\calA\cup\{\psi_{j}:b_j=1\})\setminus\{\phi\},									%				\calB\cup\{\psi_j:b_j=0\},\calC,\calD$}\;
%			}
%			\uElseIf{$\phi =S_i(\psi_1, \ldots ,\psi_k)$, $\phi \in\calB$, and $\bin{i}=c_1\ldots c_m$}{
%					\Return $\bigvee_{(b_1, \ldots ,b_k)\not\in B_i}$\KwFn{$\calA\cup\{\psi_{j}:b_j=1\},							(\calB\cup\{\psi_j:b_j=0\})\setminus \{\phi\},\calC,\calD$}\;
 % 			}			
%\end{algorithm}
%\caption{Lines 15-19 of Algorithm 1
%\ref{alg:mlsat}
% on \KwFn{$\calA,\calB,\calC,\calD$}
% where, for $b_1\ldots b_k\in \{0,1\}^*$, \\$(b_1, \ldots ,b_k)\in B_i:\Leftrightarrow\bin{i}^\frown \#b_1\ldots b_k\in A$.
% }
%\end{figure}
%Analogously to the base algorithm in \cite{Ladner77}, Algorithm \ref{alg:mlsat} requires only $\bigO{n^2}$ space (see Appendix for a detailed analysis).
\end{proof}

Using Lemmata \ref{lem:apu} and \ref{lem:alg} we can now show the $\coNEXPTIME^{\NP}$ upper bound. In the proof we utilize the following connection between alternating Turing machines and the exponential time hierarchy at the level $\coNEXPTIME^{\NP}=\PiE{2}$.
\begin{thmC}[\cite{ChandraKS81,Luck16}]\label{thm:alternation}
$\SigmaE{k}$ (or $\PiE{k}$) is the class of problems recognizable in exponential time by an alternating Turing machine which starts in an existential (universal) state and alternates at most $k-1$ many times.
\end{thmC}
%Recall that $\SigmaE{k}$ and $\PiE{k}$ are the $k$th levels of the exponential hierarchy, defined by $\SigmaE{0}:=\PiE{0}:=\EXPTIME$, and for $k\geq 1$ recursively by
%$\SigmaE{k}:=\NEXPTIME^{\SigmaP{k-1}}$ and $\PiE{k}:=\coNEXPTIME^{\SigmaP{k-1}}$. 
%Recall also that $\SigmaP{0}:=\PiP{0}:=\PTIME$, and for $k\geq 1$,
%$\SigmaP{k}:=\NP^{\SigmaP{k-1}}$ and $\PiP{k}:=\coNP^{\SigmaP{k-1}}$. In particular, observe that $\coNEXPTIME^{\NP} = \PiE{2}$.

%For the definition of an alternating Turing machine we refer the reader to \cite{ChandraKS81}.

%\begin{exa}[\cite{HLKV16}]The formula $\forall \tuple x (\U y \exists z) \neg y \leftrightarrow z$ under the constraint $(\{\tuple x\},\{\tuple x\})$ expresses that every $|\tuple x|$-ary Boolean function has a negation.
%\end{exa}

%%UPPER BOUND
\begin{thm}\label{thm:entail_mldepup}
The entailment problem for $\EMLDep$ is in $\coNEXPTIME^{\NP}$.
\end{thm}
\begin{proof}
Assuming an input  $\phi_1, \ldots ,\phi_n$ of $\EMLDep$ formulae, we show how to decide in $\PiE{2}$ whether $\{\phi_1, \ldots ,\phi_{n-1} \}\models \phi_n$. By Theorem \ref{thm:alternation} it suffices to construct an alternating exponential-time algorithm that switches once from an universal to an existential state.
By Lemma \ref{lem:apu}, $\{\phi_1, \ldots ,\phi_{n-1} \}\models \phi_n$ iff for all $\tuple f_1, \ldots ,\tuple f_{n-1}$ there is $\tuple f_n$ such that 
\begin{equation}\label{eqw}
\{  \phi_1(\tuple f_1/\tuple d_1), \ldots , \phi_{n-1}(\tuple f_{n-1}/\tuple d_{n-1})\} \models \phi(\tuple f_n/\tuple d_n).
\end{equation} 
Recall from the proof of Lemma \ref{lem:apu} that all the formulae in \eqref{eqw} belong to $\ML$. Hence by the flatness property (Proposition \ref{prop:ml_flatness})   $\models$ is interchangeable with $\models_{\ML}$ in \eqref{eqw}. It follows that $\eqref{eqw}$ holds iff 
\begin{equation}\label{eqx}
\phi:=\phi_1(\tuple f_1/\tuple d_1)\wedge \ldots \wedge \phi_{n-1}(\tuple f_{n-1}/\tuple d_{n-1})\wedge \neg \phi(\tuple f_n/\tuple d_n)
\end{equation} 
is not satisfiable with respect to the standard Kripke semantics of modal logic. 
%
%Recall that each $\tuple d_i$ is a list %$\bigl(\phi_{i_1}:=\dep{\tuple \alpha_{i_1},\beta_{i_{1}}}, \ldots ,\phi_{i_{m_i}}:=\dep{\tuple \alpha_{i_{m_i}},\beta_{i_{m_i}}}\bigr)$ %
%of all dependence atom subformulae that appear in $\phi_i$, and $\tuple f_i$ is a selected list of corresponding witness functions. % $f_{i_1}, \ldots ,f_{i_{m_i}}$, each $f_{i_j}$ being from $\{\top,\bot\}^{|\tuple \alpha_{i_j}|}$ to $\{0,1\}$. 
By the equivalence in \eqref{eqsuc} we notice that \eqref{eqx} is not satisfiable with respect to $\models_{\ML}$ iff $\phi^*$ is not satisfiable over the selected functions with respect to $\models_{\RML}$, where $\phi^*$ is obtained from $\phi$ by replacing each $D(f,\tuple \alpha,\beta)$  of the form \eqref{eqdep} with the predicate $f(\tuple \alpha)=\beta$. % and  each appearance of $\neg $, $\Diamond$, or $\psi_0 \vee \psi_1$ respectively with $\negg$, $\negg \Box \negg$, or $\negg(\negg \psi_0 \wedge \negg \psi_1)$. 
 The crucial point here is that $\phi^*$ is only of length $\bigO{n\log n}$ in the input. 

The algorithm now proceeds as follows. The first step is to universally guess functions listed in $\tuple f_1 \ldots \tuple f_{n-1}$, followed by an existential guess over functions listed in $\tuple f_n$. 
The next step is to transform the input to the described $\RML$ formula $\phi^*$. %The crucial point here is that $\phi^*$ is only polynomial in the length of the input.
 The last step is to run Algorithm \ref{alg:mlsat} on $\texttt{Sat}(\phi^*,\emptyset,\emptyset,\emptyset)$ replacing queries to the oracle with investigations on the guessed functions, and return true iff the algorithm returns false. By Lemma \ref{lem:alg}, Algorithm \ref{alg:mlsat} returns false iff \eqref{eqw} holds over the selected functions. Hence, by Lemma \ref{lem:apu} we conclude that the overall algorithm returns true iff $\{\phi_1, \ldots ,\phi_{n-1} \}\models \phi_n$.

Note that this procedure involves polynomially many guesses, each of at most exponential length. Also, Algorithm \ref{alg:mlsat} runs in exponential time and thus each of its implementations has at most exponentially many oracle queries. Hence, we conclude that the given procedure decides $\EMLDep$-entailment  in $\coNEXPTIME^{\NP}$. 
\end{proof}
Notice that the decision procedure for $\models \phi$ does not involve any universal guessing.
Therefore, we obtain immediately a $\NEXPTIME$ upper bound for the validity problem of $\EMLDep$.
\begin{cor}\label{cor:val_mldepup}
The validity problem for $\EMLDep$ is in $\NEXPTIME$.
\end{cor}

\section{Propositional Dependence Logics}\label{sect:prop}
Before showing that $\coNEXPTIME^{\NP}$ is also the lower bound for the entailment problem of the propositional fragment of $\MLDep$, we need to formally define this fragment. We also need to present other propositional variants that will be examined later in this article. Beside extending our investigations to propositional independence and inclusion logics, we will also study the extensions of these logics with additional universal and existential quantifiers. This section is divided into two subsections. Sect \ref{subsect:propintro} introduces different variants of propositional dependence logic.  Sect \ref{subsect:tomdl} shows that decision problems for quantified propositional dependence logics can be reduced to the same problems over modal dependence logics.
\subsection{Introduction to Propositional Dependence Logics}\label{subsect:propintro}
The syntax of \emph{propositional logic} ($\PL$) is generated by the following grammar:
\begin{equation}\label{def:qprop}
\phi\ddfn p \mid \neg p \mid (\phi \wedge \phi) \mid (\phi \vee \phi)
\end{equation}
The syntaxes of propositional dependence, independence, and inclusion logics ($\PLDep$, $\PLInd$, $\PLInc$, resp.) are obtained by extending the syntax of $\PL$ with dependence atoms of the form \eqref{depatom}, independence atoms of the form \eqref{indatom}, and inclusion atoms of the form \eqref{incatom}, respectively. Furthermore, the syntax of $\PLDis$ extends \eqref{def:qprop} with the grammar rule $\phi\ddfn \phi\cvee \phi$.

The formulae of these logics are evaluated against propositional teams. Let $V$ be a set of variables. We say that a function $s\colon V\to \{0,1\}$ is a \emph{(propositional) assignment} over $V$, and a \emph{(propositional) team} $X$ over $V$ is a set of propositional assignments over $V$.
 A team $X$ over $V$ induces  a Kripke model $\calM_X=(T_X,\emptyset, \pi)$ where $T_X=\{w_s\mid s\in X\}$  and $w_s \in \pi(p)\Leftrightarrow s(p)=1$ for $s\in X$ and $p\in V$.
The team semantics for propositional formulae is now defined as follows:
\[X\models \phi :\Leftrightarrow \calM_X,T_X\models \phi,\]
where $\calM_X,T_X\models \phi$ refers to the team semantics of modal formulae (see Sect. \ref{sect:mdl}). If $\phi^*$ is a formula obtained from $\phi$ by replacing all propositional atoms $p$ (except those inside a dependence atom)  with predicates $A(p)$, then we can alternatively describe that $X\models \phi$ iff $ \calM=(\{0,1\},A:=\{1\})$ and $X$ satisfy $\phi^*$ under the lax team semantics of first-order dependence logics \cite{galliani12}.

%We will also examine validity and entailment in quantified propositional dependence logic which is a team semantics adaptation and generalization of the dependency quantified Boolean formula problem \cite{HLKV16}. This problem, shown to be $\NEXPTIME$-complete in \cite{Peterson2001}, extends the quantified Boolean formula problem, the standard $\PSPACE$-complete problem, with the introduction of additional quantification constraints.  To this end, we start with the introduction of quantified propositional logic. The syntax 
\emph{Quantified propositional logic} ($\QPL$) is obtained by extending that of $\PL$ with universal and existential quantification over propositional variables. Their semantics is given in terms of so-called duplication and supplementation teams. Let $p$ be a propositional variable and $s$ an assignment over $V$. We denote by $s(a/p)$ the assignment over $V\cup\{p\}$ that agrees with $s$ everywhere, except that it maps $p$ to $a$.  Universal quantification of a propositional variable $p$ is defined in terms of \emph{duplication teams} $X[\{0,1\}/p]:=\{s(a/p)\mid s\in X,a\in \{0,1\}\}$ that extend teams $X$ with all possible valuations for $p$. Existential quantification is defined in terms of \emph{supplementation teams} $X[F/p]:=\{s(a/p)\mid s\in X, a\in F(s)\}$ where $F$ is a mapping from $X$ into $\{\{0\},\{1\},\{0,1\}\}$. The supplementation team $X[F/p]$ extends each assignment of $X$ with a non-empty set of values for $p$. The satisfaction relations $X\models \exists p \phi$ and $X\models \forall p \phi$ are now given as follows:
\begin{align*}
  X\models \exists p \phi \quad&:\Leftrightarrow \quad \exists F\in {}^X\{\{0\},\{1\},\{0,1\}\}: X[F/p]\models \phi ,\\
    X\models \forall p \phi \quad&:\Leftrightarrow \quad X[\{0,1\}/p]\models \phi.
\end{align*}

We denote by $\QPLDep$  the extension of $\PLDep$ with quantifiers and define $\QPLInd$, $\QPLInc$, and $\QPLDis$ analogously.  Observe that the flatness and downward closure properties of modal formulae  (Propositions \ref{prop:ml_flatness} and \ref{prop:ml_dc}, resp.) apply now analogously to propositional formulae. We also have that $\QPLInc$ is closed under taking unions of teams. Note that  $\models_{\PL}$ here refers to the standard semantics of propositional logic. Furthermore,  $\Var{\phi}$  refers to the set of variables appearing in $\phi$, and  $\Fr{\phi}$ to the set of free variables appearing in a formula $\phi$, both defined in the standard way. 
 Sometimes we write $\phi(p_1, \ldots ,p_n)$ instead of $\phi$ to emphasize that $\Fr{\phi}=\{p_1, \ldots, p_n\}$.  
\begin{prop}[Flatness \cite{vaananen07}]\label{prop:qpl_flatness}
Let $\phi$ be a formula in $\QPL$, and let $X$ be a team over $V\supseteq \Fr{\phi}$. Then:
\begin{align*}
X\models \phi \quad \Leftrightarrow \quad & \forall s\in X: s\models_{\PL} \phi.
\end{align*}
\end{prop}
\begin{prop}[Downward Closure \cite{vaananen07}]\label{prop:qpl_dc}
Let $\phi$ be a formula in $\QPLDep$ or $\QPLDis$, and let $X$ be a team over a set $V\supseteq \Fr{\phi}$ of propositional variables. Then:
\begin{align*}
Y \sub X\textrm{ and }X\models \phi \quad\Rightarrow \quad& Y\models \phi.
\end{align*}
\end{prop}
\begin{prop}[Union Closure \cite{galliani12}]\label{prop:qpl_uc}
Let $\phi$ be a formula in $\QPLInc$, and let $X$ and $Y$ be teams over a set $V\supseteq \Fr{\phi}$ of propositional variables. Then:
\begin{align*}
X\models \phi \textrm{ and }Y\models \phi \quad\Rightarrow \quad& X\cup Y\models \phi.
\end{align*}
\end{prop}
We  denote the \emph{restriction} of an assignment $s$ to variables in $V$ by $s\upharpoonright V$, and define the restriction of a team $X$ to $V$, written $X\upharpoonright V$, as $\{s\upharpoonright V\mid s\in X\}$.  We conclude this section by noting that, similar to the first-order case, quantified propositional dependence logic satisfies the following {locality} property.
\begin{prop}[Locality \cite{vaananen07}]\label{prop:qpl_locality}
Let $\phi$ be a formula in $\calL$ where $\calL\in\{\QPLDep,\QPLInd,\QPLInc,$\\$\QPLDis\}$, let $X$ be a team over $V\supseteq \Fr{\phi}$, and let $V'$ be such that $\Fr{\phi}\sub V'\sub V$. Then:
\begin{align*}
X\models \phi \quad \Leftrightarrow \quad & X\upharpoonright V'\models \phi.
\end{align*}
\end{prop} 

\subsection{Reductions from Quantified Propositional to Modal Logics}\label{subsect:tomdl}
In this section we show how to generate simple polynomial-time reductions from quantified propositional dependence logics to  modal dependence logics with respect to their entailment and validity problem. %Assume we are given a set of formulae $\Sigma\cup\{\phi\}$ from quantified dependence, independence, or inclusion logic. We now show how to construct in polynomial time a set of formulae $\Sigma'\cup\{\phi'\}$ from corresponding modal dependence logics such that $\Sigma\models \phi$ iff $\Sigma'\models \phi'$. The first step is to transform the formulae in $\Sigma \cup\{\phi\}$ to prenex normal form.
First we present Lemma \ref{normal form} which is a direct consequence of \cite[Lemma 14]{galhankon13} that presents prenex normal form translations in the first-order dependence logic setting  over structures with universe size at least $2$. The result follows by the obvious first-order interpretation of quantified propositional formulae: satisfaction of a quantified  propositional formula $\phi$ by a binary team $X$ can be replaced with satisfaction of $\phi^*$ by $\calM:=(\{0,1\},P^{\calM}:=\{1\})$ and $X$, where $\phi^*$ is a formula obtained from $\phi$ by replacing atomic propositional formulae $p$ and $\neg p$ respectively with $P(p)$ and $\neg P(p)$.

\begin{lemC}[\cite{galhankon13}]\label{normal form}
Any formula $\phi$ in $\LL$, where $\LL\in \{\QPLDep,\QPLInc,\QPLInd\}$, is logically equivalent to a polynomial size formula $Q_1 p_1 \ldots Q_n p_n\psi$ in $\LL$ where $\psi$ is quantifier-free and $Q_i\in\{\exists, \forall \}$ for $i=1, \ldots ,n$.
%Normal form for all QPL(D) logics: First quantifier prefix, each variable quantified at most once, and then qunatifier free part.
\end{lemC}
 %Assume now a formula of the form $Q_1 p_1 \ldots Q_k p_k\psi$. 
 Next we show how to describe in modal terms a quantifier block $Q_1 p_1 \ldots Q_n p_n$.  Using the standard method in modal logic we construct a formula $\tree{V,p_1, \ldots ,p_n}$ that enforces the complete binary assignment tree over $p_1, \ldots ,p_n$ for a team over  $V$ where $\{p_1,\ldots ,p_n\}$ and $V$  are disjoint \cite{Ladner77}. The formulation of $\tree{V,p_1, \ldots ,p_n}$ follows the presentation in \cite{HLKV16}. %Branching in the tree will correspond to quantification of proposition variables.
%
%The binary tree is forced in the standard way by modal formulae:
%
%
 We define $\Store{p}{n}\dfn (p\land\Box^n p)\lor(\lnot p\land\Box^n\lnot p)$, where $\Box^n$ is a shorthand for $\overbrace{\Box\cdots\Box}^{n\text{ many}}$, to impose the existing values  for $p$ to successors in the tree.
We also define $\Branch{p}{n}\dfn \Diamond p\land\Diamond\lnot p \land \Box \hspace{.3mm}\Store{p}{n}$ to  indicate that there are $\ge 2$ successor states which disagree on the variable $p$ and that all successor states preserve their values up to branches of length $n$.
%h
Then we let
\begin{align*}
\tree{V,p_1, \ldots ,p_n}\dfn  \bigwedge_{q\in V}\bigwedge_{i=1}^{n} \Store{q}{i} \wedge \bigwedge_{i=0}^{n-1}\Box^i \Branch{p_{i+1}}{n-(i+1)}.
\end{align*}
% is the $i$-times concatenation of $\Box$.
 %The formula  $\tree{p,n}$ forces a complete binary assignment tree of depth $n$ for proposition symbols $p_1,\dots,p_n$.
 Note that $\tree{V,p_1, \ldots ,p_n}$ is an $\ML$ formula and hence has the flatness property by Proposition \ref{prop:ml_flatness}.

 %When $\varphi$ is a $\QPLInc$-formula, we denote by $\varphi'$ the $\MInc$ formula that is obtained from  $\varphi$ by substituting each existential quantifier $\exists p$ by $\Diamond$ and each universal quantifier $\forall p$ by $\Box$.
\begin{restatable}{thm}{aseit}
\label{thm:reduction}
The satisfiability, validity, and entailment problems  for $\QPLDep$, $\QPLInd$, or $\QPLInc$ are polynomial-time reducible to the satisfiability, validity, and entailment problems for $\MLDep$, $\MLInd$, or $\MLInc$, respectively.
\end{restatable}

\begin{proof}
Consider first the entailment problem, and assume that $\Sigma\cup\{\phi\}$ is a finite set of formulae in either $\QPLDep$, $\QPLInd$, or $\QPLInc$.
By Lemma \ref{normal form} each formula in $\theta\in \Sigma \cup\{\phi\}$ can be transformed in polynomial time to the form $\theta_0=Q_1 p_1 \ldots Q_n p_n\psi$ where $\psi$ is quantifier-free. %Moreover, by the locality principle (Proposition \ref{prop:qpl_locality}) 
 We may assume that the variable sequences  $p_1, \ldots, p_n$ corresponding to these quantifier blocks are  initial segments of a shared infinite list $p_1, p_2,p_3, \ldots $ of variables.
Assume $m$ is the maximal length of the quantifier blocks that appear in any of the translations, and let $V$ be the set of variables that appear free in some of them. \Wloss we may assume that $\{p_1, \ldots ,p_m\}$ and $V$ are disjoint. We  let $\theta_1$ be obtained from $\theta_0$ by replacing quantifiers $\exists$ and $\forall$ respectively with $\Diamond$ and $ \Box$.  It follows that $\Sigma \models \phi$ iff $\{\theta_1\mid \theta\in \Sigma\}\cup\{\tree{V,p_1, \ldots ,p_n}\}\models \phi_1$.\footnote{Notice that the direction from left to right does not hold  under the so-called strict team semantics where $\exists$ and $\Diamond$ range over individuals.  These two logics are not downwards closed and  the modal translation does not prevent the complete binary tree of having two distinct roots that agree on the variables in $V$. The same proviso applies to the translations of validity and satisfiability.}
For the validity problem, we observe that $\models \phi$ iff $\models \tree{V,p_1, \ldots ,p_n}^\bot \vee (\tree{V,p_1, \ldots ,p_n}\wedge \phi_1)$. Intuitively the formula on the right states that excluding those worlds which do not encode a complete binary tree that preserves $V$, the resulting worlds form a team that satisfies the encoding  $\phi_1$  of $\phi$. Furthermore, for the satisfiability problem we have that $\phi$ is satisfiable iff $\tree{V,p_1, \ldots ,p_n}\wedge \phi_1$ is. Since the reductions are clearly polynomial, this concludes the proof.
\end{proof}

\section{Lower Bound for $\PLDep$ Entailment}\label{sect:lower}
In this section we prove that the entailment problem  for $\PLDep$ is $\coNEXPTIME^{\NP}$-hard. This result is obtained by reducing from a variant of the  dependency quantified Boolean formula problem, which is, a complete problem for $\NEXPTIME$ extending the standard quantified Boolean formula problem by constraints.  The variant we present was introduced in \cite{HLKV16}, and it gives complete problems for different levels of the exponential  hierarchy. The following presentation follows mainly \cite{HLKV16}, deviating only in notation.

\begin{defi}\label{apudef}%[Dependency Quantified Boolean Formula]
%A  tuple $\tuple C=(
%\tuple c_1, \ldots ,\tuple c_m)$ where $\tuple c_i$ are tuples of propositional variables from $\{p_1,\ldots ,p_n\}$ is called a \emph{constraint} over $\{p_1,\ldots ,p_n\}$. 
 A \emph{dependency quantified Boolean formula} (DQBF) is a pair $(\phi,\tuple C)$ where $\tuple C=(
\tuple c_1, \ldots ,\tuple c_m)$ is a list of sequences of propositional variables from  $\{p_1, \ldots ,p_n\}$, and $\phi$ is a  formula of the form
\[ \exists f_1 \ldots \exists f_m \forall p_1 \ldots \forall p_n \theta\]
where $p_1, \ldots ,p_n$ are pairwise disjoint, $f_1, \ldots ,f_m$ are pairwise disjoint functional variables, and $\theta$ is a quantifier-free propositional formula in which only the quantified variables $p_i$ and functional variables $f_i$ with arguments $\tuple c_i$ may appear.  The DQBF instance $(\phi,\tuple C)$ is true if there are functions $f_i\colon \{0,1\}^{\lvert \tuple c_i\rvert}\to \{0,1\}$ such that $\theta$ is true for each assignment $s
\colon \{p_1, \ldots ,p_n\} \to \{0,1\}$. 
\end{defi}
%A DQBF formula $ \exists f_1 \ldots \exists f_m \forall p_1 \ldots \forall p_n\theta$ is 
The true  quantified Boolean formula problem, $\TRUE(\DQBF)$, is now the problem of determining the truth value of a given DQBF instance.

\begin{thmC}[\cite{Peterson2001}]
$\TRUE(\DQBF)$ is \NEXPTIME-complete.
\end{thmC}
We now generalize $\TRUE(\DQBF)$  by introducing alternation to the quantifier block for functions.\footnote{Apart from notational differences, the following definition is from \cite{HLKV16}. Here we use explicitly functional variables instead of proposition variables for quantification of functions. Furthermore, to improve the correspondence between syntax and semantics, the quantifier block $\forall p_1 \ldots \forall p_n$ for propositional variables is here written after the quantifier block for functional variables.}

%These variants are discussed in Sect. \ref{subsect:qbf} and in Sect. ?????? we prove the hardness result.
%\subsection{Generalized Quantified Boolean Formulae}\label{subsect:qbf}
 %The quantified Boolean formula problem is the standard $\PSPACE$-complete problem. %These generalizations will be  later used to show several lower bound results for the entailment problems of various propositional and modal dependence logics. %The basis for these problems is the dependency quantified Boolean formula problem, obtained by introducing additional constraints for quantification in Boolean formulae \cite{Peterson2001}.\
%We introduce a recent generalization the quantified Boolean formula problem in terms of alternating blocks of second-order quantifiers.

\begin{defiC}[\cite{HLKV16}]
A \emph{$\Sigma_k$-alternating dependency quantified Boolean formula} ($\SigmaQBF{k}$) is a pair $(\phi,\calC)$ where $\phi$ is an expression of the form %a closed alternating qBf of the type
\begin{align*}
&\phi \dfn  (\exists f^1_1\ldots  \exists f^1_{j_1}) \, (\forall f^2_1 \ldots \forall f^2_{j_2}) \,  (\exists f^3_1 \ldots \exists f^3_{j_3})\,\quad \ldots\quad (\, Q f^k_1 \ldots Q f^k_{j_k})\,\forall p_1 \ldots  \forall p_{n}  \,   \theta, 
\end{align*}
where $Q\in\{\exists,\forall\}$, $\calC=(\tuple c^1_1,\ldots,\tuple c^k_{j_k})$  is a list of sequences of  propositional variables from  $\{p_1,\dots,p_n\}$, and $\theta$ is a quantifier-free propositional formula in which only the quantified variables $p_i$ and functional variables $f^i_j$ with arguments $\tuple c^i_j$ may appear.
Analogously, a \emph{$\Pi_k$-alternating dependency quantified Boolean formula} ($\PiQBF{k}$) is a pair $(\phi,\calC)$ where $\phi$ is an expression of the form %closed alternating qBf of the type
\begin{align*}
&\phi \dfn  (\forall f^1_1\ldots  \forall f^1_{j_1}) \, (\exists f^2_1 \ldots \exists f^2_{j_2}) \,  (\forall f^3_1 \ldots \forall f^3_{j_3})\,\quad \ldots\quad (\, Q f^k_1 \ldots Q f^k_{j_k})\,\forall p_1 \ldots  \forall p_{n}  \,   \theta.
\end{align*}
%A \emph{simple alternating qBf} is a simple $\Sigma_i$-alternating or  $\Pi_i$-alternating qBf for some $i$.
The sequence $\calC$ is called the \emph{constraint} of $\phi$.
\end{defiC}

%Notice that $\TRUE(\Sigma_1\text{-}\ADQBF)$ corresponds to the dependency Boolean formula problem. 

%%\end{defi}
%$\SigmaQBF{k}$ and $\PiQBF{k}$  now generalize $\DQBF$ as follows. 
The truth condition of a $\SigmaQBF{k}$ or a $\PiQBF{k}$ instance is defined by a generalization of  that in Definition \ref{apudef} such that each $Q f^i_j$ where $Q\in \{\exists, \forall\}$ is interpreted as existential/universal quantification over functions $f^i_j\colon \{0,1\}^{\lvert \tuple c^i_j\rvert}\to \{0,1\}$.  Let us now denote the associated decision problems by $\TRUE(\Sigma_k\text{-}\ADQBF)$ and $\TRUE(\Pi_k\text{-}\ADQBF)$. These problems  characterize levels of the exponential hierarchy in the following way. 
%Using these and alternating Turing machines (see \cite{ChandraKS81})
%we conclude this section with two theorems describing the complexity of the exponential hierarchy. 
%Notice that $\SigmaE{k}$ and $\PiE{k}$ denote the $k$th levels of the exponential hierarchy, defined by $\SigmaE{0}:=\PiE{0}:=\EXPTIME$, and for $k\geq 1$ recursively by
%$\SigmaE{k}:=\NEXPTIME^{\SigmaP{k-1}}$ and $\PiE{k}:=\coNEXPTIME^{\SigmaP{k-1}}$. %Here $\SigmaP{k}$ and $\PiP{k}$ are the $k$th levels of the polynomial hierarchy, i.e., 
%Recall that $\SigmaP{0}:=\PiP{0}:=\PTIME$, and for $k\geq 1$,
%$\SigmaP{k}:=\NP^{\SigmaP{k-1}}$ and $\PiP{k}:=\coNP^{\SigmaP{k-1}}$.
%For the definition of an alternating Turing machine we refer the reader to \cite{ChandraKS81}.
%\begin{exa}[\cite{HLKV16}]The formula $\forall \tuple x (\U y \exists z) \neg y \leftrightarrow z$ under the constraint $(\{\tuple x\},\{\tuple x\})$ expresses that every $|\tuple x|$-ary Boolean function has a negation.
%\end{exa}
\begin{thmC}[\cite{HLKV16}]\label{thm:odd-k-dqbf-hardness}
Let $k \geq 1$. For odd $k$ the problem $\TRUE(\Sigma_k\text{-}\ADQBF)$ is $\SigmaE{k}$-complete. For even $k$ the problem $\TRUE(\Pi_k\text{-}\ADQBF)$ is $\PiE{k}$-complete.
    	%The problem $\TRUE(\ADQBF)$ is $\AEXPPOLY$-complete.
\end{thmC}
Since  $\TRUE(\Pi_2\text{-}\ADQBF)$ is   $\coNEXPTIME^{\NP}$-complete, we can show the lower bound via an reduction from it.
Notice that regarding the validity problem of $\PLDep$, we already have the following lower bound.
\begin{thmC}[\cite{Virtema14}]\label{thm:jonni}
The validity problem for $\PLDep$ is $\NEXPTIME$-complete, and for $\MLDep$ and $\EMLDep$ it is $\NEXPTIME$-hard.
\end{thmC}
This result was shown by a reduction from the dependency quantified Boolean formula problem (\ie $\TRUE(\Sigma_1\text{-}\ADQBF)$) to the validity problem of $\PLDep$. We use essentially the same technique to reduce from $\TRUE(\Pi_2\text{-}\ADQBF)$ to the entailment problem of $\PLDep$. 
\begin{thm}
\label{thm:entailpldep}
The entailment problem for $\PLDep$  is  $\coNEXPTIME^{\NP}$-hard.
\end{thm}

\begin{proof}
By Theorem \ref{thm:odd-k-dqbf-hardness} it suffices to show a reduction from $\TRUE(\Pi_2\text{-}\ADQBF)$.
Let $(\phi,\calC)$ be an instance of $\Pi_2\text{-}\ADQBF$ in which case $\phi$ is of the form 
\[\forall f_1 \ldots \forall f_m \exists f_{m+1} \ldots \exists f_{m+m'}\forall p_1 \ldots \forall p_n  \theta\]
and $\calC$ lists tuples $\tuple c_i\sub \{p_1, \ldots ,p_n\}$, for $i=1, \ldots ,m+m'$. Let $q_i$ be a fresh propositional variable for each functional variable $f_i$. We define
 $\Sigma:=\{\dep{\tuple c_i,q_i}\mid i=1, \ldots ,m\}$ and  
\[\psi:=\theta' \vee \bigvee^{m+m'}_{i=m+1}\dep{\tuple c_i,q_i},\]
where $\theta'$ is obtained from $\theta$ by replacing  occurrences of $f_i(\tuple c_i)$ with $q_i$.
Clearly, $\Sigma$ and $\psi$ can be constructed from $(\phi,\calC)$ in polynomial time. It remains to show that 
 $ \Sigma \models \psi$ iff  $\phi$ is true. 
 
 Assume first that $\Sigma \models \psi$ and let $f_i\colon \{0,1\}^{|\tuple c_i|}\to\{0,1\}$ be arbitrary for $i=1, \ldots ,m$. Construct a team $X$ that consists of all assignments $s$ that map $p_1, \ldots ,p_n,q_{m+1}, \ldots ,q_{m+m'}$ into $\{0,1\}$ and $q_1, \ldots ,q_m$  respectively to $f_1(s(\tuple c_1)), \ldots ,f_m(s(\tuple c_m))$. Since $X\models \Sigma$ we find $Z,Y_1, \ldots ,Y_{m'}\sub X$ such that $Z\cup Y_1\cup \ldots \cup Y_{m'} =X$, $Z\models \theta' $, and $Y_i\models \dep{\tuple c_{m+i},q_{m+i}}$ for $i=1, \ldots ,m'$.  We may assume that each $Y_i$ is a maximal subset satisfying $\dep{\tuple c_{m+i},q_{m+i}}$, i.e., for all $s\in X\setminus Y_{i}$, $Y_i\cup\{s\}\not\models \dep{\tuple c_{m+i},q_{m+i}}$. Since $X$ takes only Boolean values, the complement $X\setminus Y_i$ is a maximal subset satisfying $\dep{\tuple c_{m+i},q_{m+i}}$,  too. By downward closure (Proposition \ref{prop:qpl_dc}) we may assume that $Z$ does not intersect any of the subsets $Y_1, \ldots ,Y_{m'}$. Consequently, $Z$ is a maximal subset that satisfies all $\dep{\tuple c_{m+i},q_{m+i}}$ for $i=1, \ldots ,m'$. It follows that there are functions $f_i\colon \{0,1\}^{|\tuple c_i|}\to \{0,1\}$, for $i=m+1, \ldots ,m+m'$, such that 
\begin{align*}Z= \{&s\big(f_{m+1}(s(\tuple c_{m+1}))/q_{m+1}, \ldots , 
  f_{m+m'}(s(\tuple c_{m+m'}))/q_{m+m'} \big)\mid s\in X\}.
\end{align*}
Notice that $Z$ is maximal with respect to $p_1, \ldots ,p_n$, i.e., $Z\upharpoonright \{p_1, \ldots ,p_n\}={}^{\{p_1, \ldots ,p_n\}}\{0,1\}$. Hence, by the flatness property (Proposition \ref{prop:qpl_flatness}), and since  $Z\models \theta'$, it follows that $\theta'$ holds for all values of $p_1, \ldots ,p_n$ and for the values of $q_1, \ldots ,q_{m+m'}$ chosen respectively according to $f_1, \ldots ,f_{m+m'}$. Therefore, $\phi$ is true which shows the direction from left to right.
 
 Assume then that $\phi$ is true, and let $X$ be a team satisfying $\Sigma$.  Then there are functions $f_i\colon \{0,1\}^{|\tuple c_i|}\to \{0,1\}$ such that $f(s(\tuple c_i))=s(q_i)$ for $s \in X$ and $i=1, \ldots ,m$. Since $\phi$ is true we find functions $f_i\colon \{0,1\}^{|\tuple c_i|}\to \{0,1\}$, for $i=m+1, \ldots ,m+m'$, such that for all $s\in X$:
\begin{align}\label{eq4}
s\big(f_{m+1}(s(\tuple c_{m+1}))/q_{m+1}, \ldots ,
f_{m+m'}(s(\tuple c_{m+m'}))/q_{m+m'}\big)\models \theta'.
\end{align}
Again, since $X$ is Boolean, it follows that $Y_i:=\{s\in X\mid s(q_i)\neq f(s(\tuple c_i))\}$ satisfies $\dep{\tuple c_i,q_i}$ for $i=m+1, \ldots ,m+m'$. Then it follows by \eqref{eq4} and flatness (Proposition \ref{prop:qpl_flatness}) that $X\setminus (Y_{m+1}\cup \ldots \cup Y_{m+m'}) $ satisfies $\theta'$. Therefore, $\Sigma \models \psi$  which concludes the direction from right to left.
\end{proof}

%%%%%%MODAL DEPENDENCE LOGIC

\section{Validity and Entailment in Modal and Propositional Dependence Logics}\label{sect:mldep}

We may now draw together the main results of Sections \ref{sect:upper} and \ref{sect:lower}. There it was shown that in terms of the entailment problem $\coNEXPTIME^{\NP}$ is both an upper bound for $\EMLDep$ and an lower bound for $\PLDep$. Therefore, we obtain in Theorem \ref{thm:entail} that for all the logics inbetween it is also the exact complexity bound. Furthermore, Theorem \ref{thm:reduction} implies that we can count $\QPLDep$ in this set of logics. 
%The proof of this uses standard reduction methods and is located in Appendix.
%By standard methods one can show that quantified propositional dependence logics reduce to modal dependence logics in terms of their entailment problems. For the proof, see Appendix.

%\begin{restatable}{thm}{aseit}
%\label{thm:reduction}
%The satisfiability, validity, and entailment problems  for  $\QPLDep$ are polynomial-time reducible to the satisfiability, validity, and entailment problems for $\MLDep$, respectively.
%\end{restatable}

\begin{thm}\label{thm:entail}
The entailment problem for $\EMLDep$, $\MLDep$, $\QPLDep$, and $\PLDep$ is \\$\coNEXPTIME^{\NP}$-complete.
\end{thm}
\begin{proof}
The upper bound   for $\EMLDep$ and $\MLDep$ was shown in Theorem \ref{thm:entail_mldepup}, and  by Theorem \ref{thm:reduction} the same upper bound applies to $\QPLDep$ and $\PLDep$. The lower bound for all of the logics comes from Theorem \ref{thm:entailpldep}.
\end{proof}
We also obtain that all the logics inbetween $\PLDep$ and $\EMLDep$ are $\NEXPTIME$-complete in terms of their validity problem. The proof arises analogously from Corollary \ref{cor:val_mldepup} and Theorem \ref{thm:jonni}.
\begin{thm}\label{cor:val_mldep}
The validity problem for $\EMLDep$, $\MLDep$, $\QPLDep$, and $\PLDep$ is   $\NEXPTIME$-complete.
\end{thm}
%\begin{proof}
%Since the algorithm in the proof of Theorem \ref{thm:entail_mldep} involves no universal branching, given an empty assumption set, we obtain a $\NEXPTIME$ upper bound for the validity problem of $\EMLDep$ and  $\MLDep$. Also, by Proposition \ref{prop:qpl_locality} and \ref{prop:qpl_dc} a formula $\phi\in\QPLDep$ is valid iff $X_{\rm comp}\models \phi$, where $X_{\rm comp}$ is the complete binary team over $\Fr{\phi}$, i.e., $X_{\rm comp}:={}^{\Fr{\phi}}\{0,1\}$. Since this team is  exponential in $\phi$, and   $\QPLDep$  model checking can be done in non-deterministic polynomial time, we obtain a $\NEXPTIME$ upper bound for $\QPLDep$ and $\PLDep$.
%The lower bound for all of the logics comes from Theorem \ref{thm:jonni}.
%\end{proof}
Recall that this close correspondence between propositional and modal dependence logics only holds with respect to their entailment and validity problems. Satisfiability of propositional dependence logic is only $\NP$-complete whereas it is $\NEXPTIME$-complete for its modal variant. It is also worth noting that the proof of Theorem \ref{thm:entail_mldepup} gives rise to an alternative proof for the $\NEXPTIME$ upper bound for $\MLDep$ (and $\EMLDep$) satisfiability, originally proved in \cite{sevenster09b}. Moreover, the technique can be succesfully applied to $\MLDis$. %Hella et al. showed in \cite{HellaLSV14}  that $\EMLDep$ is exponentially more succint than $\MLDis$. Hence, in light of Theorems \ref{thm:entail_mldep} and \ref{cor:val_mldep} 
The following theorem entails that $\MLDis$ is no more complex than the ordinary modal logic.

\begin{thm}\label{thm:mldis}
The satisfiability, validity, and entailment problems for $\MLDis$ are $\PSPACE$-complete.
\end{thm}
\begin{proof}
The lower bound follows from the Flatness property of $\ML$ (Proposition \ref{prop:ml_flatness}) and the $\PSPACE$-hardness of satisfiability and validity problems for $\ML$ \cite{Ladner77}.
%Push negation to the atomic level.
  For the upper bound, it suffices to consider  the entailment problem. The other cases are analogous. Ananlogously to the proof of Lemma \ref{lem:apu} (see also Theorem 5.2 in \cite{Virtema14}) we reduce $\MLDis$ formulae to large disjunctions with the help of appropriate witness sequences. Let $\theta$ be an $\MLDis$ formula that has $m$ $\cvee$-disjunctions. Given a sequence  $\tuple s=(s_1, \ldots ,s_m)\in \{0,1\}^{m}$ we determine top-down recursively an $\ML$ formula $\eta^{\tuple s}$ for each subformula $\eta$ of $\theta$ in the following way. We let $(\eta_0\cvee \eta_1)^{\tuple s}:= \eta^{\tuple s}_{s_i}$ where $(\eta_0\cvee \eta_1)$ is the $i$th $\cvee$-disjunction of $\theta$ in some underlying ordering. Otherwise, we let $(\Box\eta)^{\tuple s}:=\Box \eta^{\tuple s}$, $(\Diamond \eta)^{\tuple s} := \Diamond \eta^{\tuple s}$, $(\eta_0\wedge \eta_1)^{\tuple s}:= \eta_0^{\tuple s} \wedge \eta_1^{\tuple s}$,  $(\eta_0\vee \eta_1)^{\tuple s} := \eta_0^{\tuple s} \vee \eta_1^{\tuple s}$, $p^{\tuple s}:=p$, and $(\neg p)^{\tuple s}:=\neg p$. It is straightforward to show that $\theta$ is  equivalent to $\bigcvee{\tuple s} \theta^{\tuple s}$.

 % For an $\MLDis$-formula   $\theta$, denote by  $F_{\theta}$  the set of all  functions  that map subformulae  $\alpha \cvee \beta$  of $\theta$ to either $\alpha$ or $\beta$. For each $f\in F_{\theta}$,  we then denote by $\theta^f$  the formula obtained from $\theta$ by replacing each subformula of the form $\alpha \cvee \beta$ with $f(\alpha \cvee \beta)$. It is straightforward to show that $\theta$ is equivalent to $\bigcvee{f\in F_{\theta}} \theta^f.$
Let now $\phi_1, \ldots ,\phi_n$ be a sequence of $\MLDis$ formulae. Analogously to the proof of Lemma \ref{lem:apu} we can show using Proposition \ref{prop:yang} that $\{\phi_1, \ldots ,\phi_{n-1}\}\models \phi_n$ iff 
for all appropriate sequences $ \tuple s_1, \ldots , \tuple s_{n-1}$ there is $ \tuple s_n$ such that $\{ \phi_1^{\tuple s_1}, \ldots , \phi_{n-1}^{\tuple s_{n-1}}\}\models \phi_n^{\tuple s_n}$. Notice that the number of $\cvee$-disjunctions appearing in $\phi_1, \ldots ,\phi_n$ is polynomial, and hence the sequence $\tuple s_1\ldots \tuple s_n$ is polynomial. It follows that the decision procedure presented in the proof of Theorem \ref{thm:entail_mldepup} can be now implemented in polynomial space. %Hence, the entailment problem for $\MLDis$ is in $\PSPACE$. 
We immediately obtain the $\PSPACE$ upper bound for validity. For satisfiability of an $\MLDis$ formula $\phi$, notice that the equivalent disjunction $\bigcvee{\tuple s} \phi^{\tuple s}$ is satisfiable iff $\phi^{\tuple s}$ is satisfiable for some sequence $\tuple s$. Since satisfiability of an $\ML$ formula $\phi^{\tuple s}$ can be determined in $\PSPACE$, it follows that satisfiability of $\phi$ can be determined in $\PSPACE$, too. This concludes the proof.
\end{proof}
Combining the proofs of Theorem \ref{thm:entail_mldepup} and Theorem \ref{thm:mldis} we also notice that satifiability, validity, and entailment can be decided in $\PSPACE$ for $\EMLDep$ formulae whose dependence atoms are of \emph{logarithmic length}.

%
%\begin{table*}[t]
%\begin{center}
%\scalebox{.8}{
%\begin{tabular}{llll}
%\toprule
%	 &  satisfiability& validity & entailment  \\\midrule
%	 $\PL$ & $\NP$ \cite{Cook71,Levin73} & $\coNP$ \cite{Cook71,Levin73} & $\coNP$ \cite{Cook71,Levin73} \\
%	  	$\ML$ &$\PSPACE$ \cite{Ladner77} & $\PSPACE$ \cite{Ladner77} & $\PSPACE$ \cite{Ladner77} \\
%	  $\MLDis$ & $\PSPACE $ \cite{sevenster09b} & $\PSPACE$ [Thm. \ref{thm:mldis}]  & $\PSPACE$ [Thm. \ref{thm:mldis}] \\
%	$\PLDep$  & $\NP$ \cite{LohmannV13} &$\NEXPTIME$ \cite{Virtema14} & $\coNEXPTIME^{\NP}$ [Thm. \ref{thm:entail}] \\
 %	$\QPLDep,\MLDep,\EMLDep$ & $\NEXPTIME$ \cite{Peterson2001,sevenster09b}, [Thm. \ref{thm:reduction}] & $\NEXPTIME$ [Thm. \ref{cor:val_mldep}]&  $\coNEXPTIME^{\NP}$ [Thm. \ref{thm:entail}]  \\
% \bottomrule
%\end{tabular}}
%\vspace{0ex}
%\caption{Summary of results. The stated complexity classes refer to completeness results.}\label{newresults}
%\end{center}
%\vspace*{-1ex}%
%\end{table*}

%%%%MODAL INDEPENDENCE LOGIC
\section{Validity and Entailment in Modal and Quantified Propositional Independence Logics}\label{sect:mlind}
Next we turn to quantified propositional logic extended with either independence or inclusion atoms. We start in this section by considering the logic $\QPLInd$ and show that the complexity of its validity and entailment are both $\coNEXPTIME^{\NP}$-complete. The upper bound is a simple adaptation of the standard model checking algorithm for team-based logics. The lower bound is shown by reducing from $\TRUE(\Pi_2\text{-}\ADQBF)$ to the validity problem of quantified propositional logic extended with both dependence and inclusion atoms. Then using Galliani's translation \cite{galliani12} from dependence and inclusion logic to independence logic the result follows. Let us start with the upper bound result.

\begin{lem}\label{lem:qplind_upper}
The entailment problem for $\QPLInd$ is in $\coNEXPTIME^{\NP}$.
\end{lem}
  
  \begin{proof}
By Theorem \ref{thm:alternation} it suffices to describe an exponential-time alternating algorithm that  recognizes the entailment problem for $\QPLInd$ and switches once from an universal to an existential state. This algorithm utilizes the standard model checking algorithm for team-based logics (see, e.g., \cite{ebbing12}), adapted for $\QPLInd$ in Algorithm \ref{alg:qplind_entail}. The input for Algorithm \ref{alg:qplind_entail} is a propositional team $X$ and a $\QPLInd$ formula $\phi$, and the algorithm returns $1$ for $\texttt{MC}(X,\phi)$ iff $X$ satisfies $\phi$. Note that the running time for  $\texttt{MC}(X,\phi)$ is not bounded by a polynomial due to possible quantification in $\phi$. Instead, the procedure takes time $f(|X|)g(|\phi |)$ for some polynomial function $f$ and an exponential function $g$, and hence $\texttt{MC}(X,\phi)$ can be determined in $\NEXPTIME$.

Let us denote by $\texttt{co-MC}(X,\phi)$ the outcome of the non-deterministic algorithm obtained from Algorithm \ref{alg:qplind_entail} by replacing existential selections with universal ones. This algorithm returns  $0$ for $\texttt{co-MC}(X,\phi)$ iff $X$ does not satisfy $\phi$, and hence decides the complement of model checking in $\coNEXPTIME$. Given a sequence $\phi_1, \ldots ,\phi_n$ of $\QPLInd$ formulae, we can now determine whether $\{\phi_1, \ldots ,\phi_{n-1}\}$ entails $ \phi_n$ in the following way. % decides non-deterministically whether $X$ satisfies $\phi$.
%The basis of this procedure is the standard model checking algorithm for team-based logics arising directly from the definition of team semantics. For $\QPLInd$ the input is a propositional team $X$ and a $\QPLInd$-formula $\phi$, and the algorithm decides non-deterministically whether $X$ satisfies $\phi$. Notice that non-determinism is only needed when considering disjunctive or existentially quantified formulae. For the details of the algorithm, we refer the reader to \cite{ebbing12}. Let us denote this model checking procedure for $X$ and $\phi$ by $\texttt{MC}(X,\phi)$.  Notice that the running time of $\texttt{MC}(X,\phi)$ is not bounded by a polynomial due to possible quantification in $\phi$. Instead, the procedure takes time $f(|X|)g(|\phi |)$ for some polynomial function $f$ and an exponential function $g$.
%Let us denote by $\texttt{MC$^*$}(X,\phi)$ the non-deterministic algorithm obtained from $\texttt{MC}(X,\phi)$ by replacing existential guesses with universal ones.
%Based on $\texttt{MC}(X,\phi)$ and $\texttt{MC$^*$}(X,\phi)$ we now present the decision procedure for $\QPLInd$-entailment. Assume we are given a sequence of $\QPLInd$-formulae $\phi_1, \ldots ,\phi_n$, and the question is to determine whether $\{\phi_1, \ldots ,\phi_{n-1}\}\models \phi_n$. The procedure runs as follows. 
 First universally guess a team $X$ over variables that occur free in some $\phi_1, \ldots ,\phi_n$. Then for $i=1, \ldots ,n-1$, %run $\texttt{MC$^*$}(X,\phi_i)$. If 
 if $\texttt{co-MC}(X,\phi_i)$ is $0$, then return true. If $\texttt{co-MC}(X,\phi_i)$ is $1$ and $i< n-1$, then move to $i+1$. Otherwise, if $\texttt{co-MC}(X,\phi_i)$ is $1$ and $i= n-1$, then switch to existential state and return true iff  $\texttt{MC}(X,\phi_n)$ is $1$.  It is straightforward to check that the described algorithm returns true iff $\{\phi_1, \ldots ,\phi_{n-1}\}\models \phi_n$. %Also, notice that the universally guessed team $X$ has possibly exponential size, and the 
Since the algorithm alternates once from  universal  to  existential state and runs in %Since $\texttt{MC}(X,\phi)$ and $\texttt{co-MC}(X,\phi)$ both run 
in exponential time, % $f(|X|)g(|\phi |)$, for some polynomial function $f$ and an exponential function $g$,  
 it follows that the procedure is in $\PiE{2}$.
\end{proof}

 \begin{algorithm}\label{alg:qplind_entail}
 \caption{A non-deterministic model checking algorithm for $\QPLInd$}
  \SetAlgoLined
  \LinesNumbered
  \SetKwInOut{Input}{Input}\SetKwInOut{Output}{Output}
  \SetKwFunction{KwFn}{MC}
  \Input{$(X,\phi)$ where $X$ is a propositional team over $\Fr{\phi}$ and $\phi\in \QPLInd$}
  \Output{\KwFn{$X,\phi$}}
  \BlankLine		
  		\uIf{$\phi = \exists p \psi$}{
			\textbf{existentially choose} $F:X\to \{\{0,1\},\{0\},\{1\}\}$ and
  			\Return \KwFn{$X[F/p],\psi$}\;
		}
		\uElseIf{$\phi = \forall p \psi$}{
				\Return \KwFn{$X[\{0,1\}/p],\psi$}\;		
		}
		\uElseIf{$\phi = \psi \vee \theta$}{
			\textbf{existentially choose} $Y,Z\sub X$ such that $Y\cup Z =X$ and
			\Return  \KwFn{$Y,\psi$} $\wedge$ \KwFn{$Z,\theta$}\;
		}
		\uElseIf{$\phi = \psi \wedge \theta$}{
			\Return  \KwFn{$X,\psi$} $\wedge$ \KwFn{$X,\theta$}\;
		}
		\uElseIf{$\phi$ is an atom}{
			\Return  $1$ iff $X\models \phi$\;
			}
	%	\uElseIf{$\phi = \neg p$}{
%			\Return  $1$ iff $s(p)=0$ for all $s\in X$\;
%			}
%		\ElseIf{$\indep\tuple p \sub \tuple q$}{
%			$Y\gets X$\;
%			\While{$Y\neq \{s\in Y \mid s(\tuple p)\sub Y(\tuple q)\}$}{
%				$Y\gets \{s\in Y \mid s(\tuple p)\sub Y(\tuple q)\}$\;
%			}
%			\Return $Y$
%		}	
\end{algorithm}

For the lower bound, we apply the fact that dependence atoms as well as inclusion atoms can be defined in independence logic. A translation for inclusion atoms can be given as follows.
\begin{thmC}[\cite{galliani12}]
The inclusion atom $\tuple p\sub \tuple q$ is equivalent to
\begin{align*}
\phi:=&\forall v_1\forall v_2 \forall \tuple r\big ( (\tuple r \neq \tuple p\wedge \tuple r \neq \tuple q)\vee ( v_1\neq v_2\wedge \tuple r\neq \tuple q) \vee \\
&((v_1=v_2\vee\tuple r=\tuple q)\wedge \indep{\emptyset}{\tuple r}{v_1v_2})\big ).
\end{align*}
\end{thmC}
The above theorem was shown in the first-order inclusion and independence logic setting but can be applied to the quantified propositional setting too since $\phi$ and $\tuple p \sub \tuple q$ are satisfied by a binary team $X$ in the quantified propositional setting iff they are satisfied by $X$ and the structure $\{0,1\}$ in the first-order setting. The lower bound can be now shown by a reduction from  $\TRUE(\Pi_2\text{-}\ADQBF)$ to validity of quantified propositional logic extended with dependence and inclusion atoms. %The proof is similar to that of Theorem \ref{thm:entailpldep}.

\begin{restatable}{lem}{akasi}
\label{lem:qplind_lower}
The validity problem for $\QPLInd$ is $\coNEXPTIME^{\NP}$-hard.
\end{restatable}
\begin{proof}
We reduce from $\TRUE(\Pi_2\text{-}\ADQBF)$ which is $\coNEXPTIME^{\NP}$-hard by Theorem \ref{thm:odd-k-dqbf-hardness}. Let $(\phi,\calC)$ be an instance of $\Pi_2\text{-}\ADQBF$. Then $\phi$ is of the form 

\[ \forall f_1 \ldots \forall f_m \exists f_{m+1} \ldots \exists f_{m+m'} \forall p_1 \ldots \forall p_n \theta\]
and $\calC$ lists tuples $\tuple c_i$ of elements from $ \{p_1, \ldots ,p_n\}$, for $i=1, \ldots ,m+m'$.
We show how to construct in polynomial time from $\phi$ a  $\QPLInd$ formula $\psi$  such that $ \psi$ is valid iff  $\phi$ is true. The free variables of $\psi$ consist of $p_i$, and fresh variables  $q_i$  that replace $f_i(\tuple c_i)$.
%Notice that all the variables $ p_1 ,\ldots , p_n, q_1 , \ldots , q_{m+m'} $ will appear free in $\psi$. 
We construct $\psi$ such that $X\models \psi$ iff one can select from $X$ a maximal subteam $Y\sub X$ that satisfies $\dep{\tuple c_i,q_i}$, for $i=1, \ldots ,m$, and furthermore from $Y$ a subteam $Z\sub Y$ that satisfies $\theta$ and $\dep{\tuple c_i,q_i}$, for $i=m+1, \ldots ,m+m'$.

First, let us construct top-down recursively formulae $\psi_i$, for $i=1, \ldots ,m$, as follows:
\begin{align}
\psi_i:=~ &\exists r_i\big(\dep{\tuple c_i q_i,r_i}\wedge \dep{\tuple c_i r_i,q_i} \wedge \label{eq3}\\
&\forall r'_i(\neg r'_i \vee (r'_i \wedge \tuple c_i r'_i\sub \tuple c_i r_i)) \wedge (\neg r_i \vee (r_i\wedge \psi_{i+1}))\big).\nonumber
\end{align}
We use $\dep{\tuple p,q}$ as a shorthand for $\indep{\tuple p}{q}{q}$. For a team $X$ \eqref{eq3} expresses %a minimal purification of a team $X$ from violations of $\dep{\tuple c_i,q_i}$, or in other words, %, for $i=1, \ldots ,m$. 
 %That is, '
% The formula in \eqref{eq3} expresses 
 that some  maximal subteam $Y\sub X$ satisfying $\dep{\tuple c_i,q_i}$ also satisfies $\psi_{i+1}$. 
The first three conjuncts of \eqref{eq3} entail that the subteam where $r_i$ takes value $1$ is a maximal subteam satisfying $\dep{\tuple c_i,q_i}$. The last conjunct indicates that this subteam  satisfies $\psi_{i+1}$. Consequently, if $X$ already satisfies $\dep{\tuple c_i,q_i}$, then  \eqref{eq3} entails that $X$ satisfies $\psi_{i+1}$. Define then
\[\psi_{m+1}:=\theta' \vee \bigvee_{i=m+1}^{m+m'} \dep{\tuple c_i,q_i},\]
where $\theta'$ is obtained from $\theta$ by replacing  occurrences of $f_i(\tuple c_i)$ with $q_i$.
This disjunction amounts to the existential selection of the functions $\tuple c_i \mapsto q_i$ for $i=m+1, \ldots ,m+m'$. 

We now claim that $\psi:= \psi_1$ is valid iff $\phi$ is true. Assume first that $\psi$ is valid, and let $f_i$ be any function from $\{0,1\}^{|\tuple c_i|}\to \{0,1\}$ for $i=1, \ldots ,m$. Let $X$ be the team that consists of all assignments $s$ that map $p_1, \ldots ,p_n, q_{m+1}, \ldots ,q_{m+m'}$ into $\{0,1\}$ and $q_1, \ldots ,q_m$ to $f(s(\tuple c_1)), \ldots, f(s(\tuple c_m))$. By the assumption $X\models \psi$. Hence, we find $F_1\colon X\to \calP(X)\setminus \{\emptyset\}$ such that
\begin{align}
X[F_1/r_1] \models~& \dep{\tuple c_1 q_1,r_1}\wedge \dep{\tuple c_1 r_1,q_1} \wedge\label{eq2} \\
&\forall r'_1(\neg r'_1 \vee (r'_1 \wedge \tuple c_1 r'_1\sub \tuple c_1 r_1)) \wedge (\neg r_1 \vee (r_1\wedge \psi_{2})).\nonumber
\end{align}
Let $X':=X[F_1/r_1][1/r'_1]$. Then $X'\models \dep{\tuple c_1, q_1}$ by the construction and $X'\models \dep{\tuple c_1 q_1, r_1}$ by \eqref{eq2}, and hence $X'\models \dep{\tuple c_1, r_1}$.  Also by the third conjunct of \eqref{eq2} $X'\models  \tuple  c_1 r'_1\sub \tuple c_1 r_1$. Therefore, it cannot be the case that that $s(r_1)=0$ for some $s\in X'$, and hence by the last conjunct of \eqref{eq2} and Proposition \eqref{prop:qpl_locality}, $X[1/r_1]\models \psi_2$. After $n$ iterations we obtain that $X[1/r_1]\ldots [1/r_n]\models \psi_{m+1}$ which implies by Proposition \eqref{prop:qpl_locality} that $X\models \psi_{m+1}$. Hence, there are $Z,Y_1, \ldots ,Y_{m'}\sub X$ such that $Z\cup Y_1\cup \ldots \cup Y_{m'} =X$, $Z\models \theta' $, and $Y_i\models \dep{\tuple c_{m+i},q_{m+i}}$ for $i=1, \ldots ,m'$. Notice that we are now at the same position as in the proof of Theorem \ref{thm:entailpldep}. Hence, we obtain that $\phi$ is true and that the direction from left to right holds.

  %We may assume that each $Y_i$ is a maximal subset satisfying $\dep{\tuple c_{m+i},q_{m+i}}$, i.e., for all $s\in X\setminus Y_{i}$, $Y_i\cup\{s\}\not\models \dep{\tuple c_{m+i},q_{m+i}}$. By Proposition \ref{prop:downward} we may assume that $Z$ does not intersect any of the subsets $Y_1, \ldots ,Y_{m'}$. It follows that there are functions $f_i:\{0,1\}^{|\tuple c_i|}\to \{0,1\}$, for $i=m+1, \ldots ,m+m'$, such that 
%\[Z= \{s[f_{m+1}(s(\tuple c_{m+1}))/q_{m+1}, \ldots ,   f_{m+m'}(s(\tuple c_{m+m'}))/q_{m+m'} ]: s\in X\}.\]
% By Proposition \ref{prop:flatness} and since  $Z\models \theta$ it follows that $\theta$ holds for all values of $p_1, \ldots ,p_n$ and values of $q_1, \ldots ,q_{m+m'}$ where the latter are chosen respectively according to $f_1, \ldots ,f_{m+m'}$. Hence, $\phi$ is true which shows the direction from left to right.

Assume then that $\phi$ is true, and let $X$ be any team whose domain contains  variables 
 $p_1, \ldots ,p_n,q_1, \ldots ,q_{m+m'}$. Considering $\psi_i$, we may choose a mapping $F_i$ so that for each value $\tuple b\in X(\tuple c_i)$, we select $l \in\{0,1\}$ such that  $\tuple bl\in X(\tuple c_iq_i)$, and set $F_i(s)=1$ iff $s(\tuple c_i q_i)= \tuple b l$. This amounts to a selection of a maximal subteam satisfying $\dep{\tuple c_i, q_i}$. 
 %either all  $s\in X$ such that $s(\tuple c_i q_i)= \tuple b 0$ or all $s'\in X$  such that $s'(\tuple c_i q_i)=\tuple b 1$ are mapped to $1$. The remaining assignments are mapped to $ 0$. 
 It follows that $X[F_i/r_i]$ satisfies the first three conjuncts of \eqref{eq3} and that  $\{s\in X[F_i/r_i]\mid F_i(s) =0\}$ satisfies $\neg r_i$. It remains to show that $\{s\in X[F_i/r_i]\mid F_i(s) =1\}$ satisfies  $\psi_{i+1}$.  Iterating this procedure we need only to show that $X':=\{s\in X[F_1/q_1]\ldots [F_n/q_n]:F_1(s)=\ldots =F_n(s_n)=1\}$ satisfies $\psi_{m+1}$. By the selection of functions $F_1, \ldots ,F_n$ we notice that $X'$ satisfies $\dep{\tuple c_i,q_i}$ for $i=1, \ldots ,n$. Again, following the proof of Theorem \ref{thm:entailpldep} we obtain that $X'$ satisfies $\theta' \vee \bigvee_{i=m+1}^{m+m'} \dep{\tuple c_i,q_i}$. This shows the direction from right to left.
Since the reduction from $\phi$ to $\psi$ can be done in polynomial time, this concludes the proof.
%Hence, there are functions $f_i:\{0,1\}^{|\tuple c_i|}\to \{0,1\}$ such that $f(s(\tuple c_i))=s(q_i)$ for $s \in X$ and $i=1, \ldots ,m$. Since $\phi$ is true we find functions $f_i:\{0,1\}^{|\tuple c_i|}\to \{0,1\}$, for $i=m+1, \ldots ,m+m'$, such that 
%\begin{equation}\label{eq4}
%\forall s\in X: s[f_{m+1}(s(\tuple c_{m+1}))/q_{m+1}, \ldots ,f_{m+m'}(s(\tuple c_{m+m'}))/q_{m+m'}]\models \theta.
%\end{equation}
%Clearly, $Y_i:=\{s\in X: s(q_i)\neq f(s(\tuple c_i))\}$ satisfies $\dep{\tuple c_i,q_i}$ for $i=m+1, \ldots ,m+m'$. Then it follows by \eqref{eq4} and by Proposition \ref{prop:flatness} that $X\setminus (Y_{m+1}\cup \ldots \cup Y_{m+m'}) $ satisfies $\theta$. Therefore, $\psi$ is valid which concludes the direction from right to left. Since the reduction from $\phi$ to $\psi$ can be done in polynomial time, this concludes the proof.
\end{proof}

The exact $\coNEXPTIME^{\NP}$ bound for $\QPLInd$ entaiment and validity follows now by Lemmata \ref{lem:qplind_upper} and \ref{lem:qplind_lower}. Theorems \ref{sdfsdf} and \ref{thm:reduction} then imply the same lower bound for $\MLInd$. This means that validity in modal independence logic is at least as hard as entailment in modal dependence logic. We leave determining the exact complexity of $\MLInd$ entailment and validity as an open question. %Based on Theorem \ref{sdfsdf} it seems reasonable to conjecture these two problems enjoy the same exact bounds.

\begin{thm}
\label{sdfsdf}
The entailment and the validity problems for $\QPLInd$ are   $\coNEXPTIME^{\NP}$-complete.
\end{thm}
\begin{cor}
\label{sdkf}
The entailment and the validity problems for $\MLInd$ are $\coNEXPTIME^{\NP}$-hard.
\end{cor}

%%%MODAL INCLUSION LOGIC
\section{Validity and Entailment in Modal and Quantified Propositional Inclusion Logics}\label{sect:mlinc}
Next we consider quantified propositional inclusion logic and show  that its  validity and entailment problems are complete for $\EXPTIME$ and $\coNEXPTIME$, respectively.  The result regarding validity is a simple observation.

%We also consider validity of $\QPLInc$ formulae and prove that it is $\EXPTIME$-complete.
%Let us also consider validity of $\QPLInc$ formulae. First we state the following theorem which follows by Theorem \ref{thm:reduction} because both $\PLInc$  and $\MLInc$ are $\EXPTIME$-complete in terms of their satisfiability problem. \cite{HellaKMV15}.
%\begin{thm}\label{thm:satqplinc}
%The satisfiability problem for  $\QPLInc$ is   $\EXPTIME$-complete.
%\end{thm}
%It is then a straightforward exercise to show that also the validity of $\QPLInc$ is $\EXPTIME$-complete.
\begin{thm}\label{thm:valqplinc}
The validity problem for  $\QPLInc$ is $\EXPTIME$-complete.
\end{thm}
\begin{proof}
For the lower bound, note that the satisfiability problem for $\PLInc$ has been shown to be $\EXPTIME$-complete in \cite{HellaKMV15}. Also, note that $\phi(\tuple p)$ is satisfiable iff $\exists \tuple p \phi(\tuple p)$ is valid. Consequently, $\QPLInc$ validity is $\EXPTIME$-hard.

For the upper bound, we notice by union closure (Proposition \ref{prop:qpl_uc}) that a formula $\phi(\tuple p)\in \QPLInc$ is valid iff it is valid over all singleton teams. Let us denote by $\phi^*$ the formula obtained from $\phi$ by replacing all inclusion atoms $\tuple q \sub\tuple r$ with $\tuple p\tuple q \sub \tuple p \tuple r$. We observe that $\phi(\tuple p)$ is valid over all singletons iff $\{\emptyset\} \models \forall \tuple p\phi^*(\tuple p)$. The direction from left to right follows by union closure and since $\tuple q \sub\tuple r$ entails $\tuple p\tuple q \sub \tuple p \tuple r$. For the direction from right to left it can be shown by induction on $\psi$ that $X\models \psi^*$ entails $\{s\in X \mid s(\tuple p) =\tuple a\}\models \psi$ for all teams $X$ and sequences of Boolean values $\tuple a$. By locality (Proposition \ref{prop:qpl_locality}) $\{\emptyset\} \models \forall \tuple p\phi^*(\tuple p)$ iff $\forall \tuple p\phi^*(\tuple p)$ is satisfiable. Since $\MLInc$ satisfiability is in $\EXPTIME$ \cite{HKMV16}, we conclude by Theorem \ref{thm:reduction} that the $\EXPTIME$ upper bound holds.
\end{proof}

Let us now prove the exact $\coNEXPTIME$ complexity bound  for $\QPLInc$ entailment. For the proof of the lower bound, we apply $\TRUE(\Sigma_1\text{-}\ADQBF)$.

\begin{lem}\label{lem:lower_qplinc}
The entailment problem for $\QPLInc$ is   $\coNEXPTIME$-hard.
\end{lem}

\begin{proof}
Theorem \ref{thm:odd-k-dqbf-hardness} states that $\TRUE(\Sigma_1\text{-}\ADQBF)$  is $\NEXPTIME$-hard. We reduce from its complement problem. Let $(\phi,\calC)$ be an instance of $\Sigma_1\text{-}\ADQBF$. Then $\phi$ is of the form $ \exists g_1 \ldots \exists g_m \forall p_1 \ldots \forall p_n \theta $ and $\calC$ lists tuples $\tuple c_i$ of elements from $\{p_1, \ldots ,p_n\}$, for $i=1, \ldots ,n$.
We show how to construct in polynomial time from $(\phi,\calC)$ two $\QPLInc$ formulae  $\psi$ and $\psi'$ such that  $(\phi,\calC)$ is false iff $\psi \models \psi'$. These formulae use fresh variables $q_i$ for encoding $g_i(\tuple c_i)$ and variables $t$ and $f$ for encoding true and false.
Denote by $\tuple p$ and $\tuple q$  sequences $p_1\ldots p_n$ and $q_1\ldots q_m$, and by $\tuple p'$ and $\tuple q'$ their distinct copies, respectively. Let  $\tuple d_i$ list the variables of $\{p_1, \ldots ,p_n\}$ that do not occur in $\tuple c_i$. %and  denote by $\tuple r_i$ the sequence $r_1\ldots r_{|\tuple d_i|}$.  
 %and let $\tuple p'$ and $\tuple q'$ denote respectively two lists of distinct copies $p'_1\ldots p'_n$ and $q'_1\ldots q'_m$. 
 Moreover, let $\tuple v_i$ be lists of fresh variables of length $|\tuple c_i|$. %, for $i=1, \ldots ,m$.  
 The idea is to describe that whenever a team $X$ is complete with respect to the values of $\tuple p$, then either one of the constraints in $\calC$ is falsified or one of the assignments of $X$ falsifies $\theta$.  We let
\begin{equation}\label{psi}
 \psi:=t\wedge \neg f \wedge \bigwedge_{i=1}^{n} (p_1\ldots p_{i-1}t\sub p_1\ldots p_{i-1}p_i \wedge p_1\ldots p_{i-1}f\sub p_1\ldots p_{i-1}p_i),
\end{equation}
%\item  $\phi_3:=\bigwedge_{i=1}^{n} (\tuple p \tuple q r_1\ldots r_{i-1}t\sub \tuple p \tuple q r_1\ldots r_{i-1}r_i \wedge \tuple p \tuple q r_1\ldots r_{i-1}f\sub \tuple p \tuple q r_1\ldots r_{i-1}r_i)$,
%\item $\phi_4:= \bigwedge_{i=1}^m \tuple c_i q_i\tuple r_i \sub \tuple c_i q_i\tuple d_i $,
 and define 
 \begin{equation}\label{psi'}
 \psi':=  \exists \tuple p' \tuple q'( \theta^{\bot}_0(\tuple p' \tuple q' /\tuple p\tuple q)\wedge \tuple p'\tuple q'\sub \tuple p\tuple q) \vee 
 \bigvee_{i=1}^m \exists \tuple v_i (\tuple v_i t\sub \tuple c_i q_i \wedge \tuple v_i f\sub \tuple c_i q_i),\end{equation}
 where $\theta_0$ is obtained from $\theta$ by replacing occurrences  of $g_i(\tuple c_i)$ with $q_i$. Above, \eqref{psi} exrpresses that a team $X$ is complete with respect to $\tuple p$. In \eqref{psi'} the first disjunct entails that $\theta_0$ is false for some assignment in $X$, and the second disjunct indicates that some constraint in $\calC$ does not hold in $X$. Furthermore, all disjunctions in \eqref{psi'} are essentially Boolean;  if a disjunct is satisfied by a non-empty subteam of $X$, it is likewise satisfied by the full team $X$. 
 We leave it to the reader to verify that $(\phi,\calC)$ is false iff $\psi \models \psi'$.   
\end{proof}

For the upper bound we refer to Algorithm \ref{alg:qplinc_entail} that was first presented in \cite{HKMV16} in the modal logic context. Given a team $X$ and a formula $\phi\in \QPLInd$, this algorithm computes deterministically the maximal subset of $X$ that satisfies $\phi$. Note that the existence of such a team is guaranteed by the union closure property of $\QPLInc$ (Proposition \ref{prop:qpl_uc}). Given an instance $\Sigma \cup \phi$ of the entailment problem, the proof idea is now to first universally guess a team $X$ (possibly of exponential size), and then check using Algorithm \ref{alg:qplinc_entail} whether $X$ is a witness of $\Sigma \not\models \phi$. Since the last part can be executed deterministically in exponential time, the $\coNEXPTIME$ upper bound follows.
%Based on it we show the $\coNEXPTIME$ upper bound for the entailment problem of quantified propositional inclusion logic. %A formula is said to be \emph{closed under unions} if its truth is preserved under taking unions of teams.

%  \begin{figure}[h]
 \begin{algorithm}\label{alg:qplinc_entail}
 \caption{A deterministic model checking algorithm for $\QPLInc$}
  \SetAlgoLined
  \LinesNumbered
  \SetKwInOut{Input}{Input}\SetKwInOut{Output}{Output}
  \SetKwFunction{KwFn}{MaxSub}
  \Input{$(X,\phi)$ where $X$ is a propositional team over $\Fr{\phi}$ and $\phi\in \QPLInc$}
  \Output{\KwFn{$X,\phi$}}
  \BlankLine		
  		\uIf{$\phi = \exists p \psi$}{
  			\Return $\{s\in X\mid s(0/p)\in $ \KwFn{$X[\{0,1\}/p],\psi$}$\textrm{ or }s(1/p)\in $ \KwFn{$X[\{0,1\}/p],\psi$}$\}$\;
		}
		\uElseIf{$\phi = \forall p \psi$}{
			$Y\gets X[\{0,1\}/p]$\;
			\While{$Y\neq \{s\in Y \mid \{s(0/p),s(1/p)\}\sub$  \KwFn{$Y,\psi$}$\}$}{
				$Y\gets \{s\in Y \mid \{s(0/p),s(1/p)\}\sub$  \KwFn{$Y,\psi$}$\}$\;
			}
			\Return $\{s\in X\mid \{s(0/p),s(1/p)\}\sub Y\}$
		}
		\uElseIf{$\phi = \psi \vee \theta$}{
			\Return  \KwFn{$X,\psi$} $\cup$ \KwFn{$X,\theta$}\;
		}
		\uElseIf{$\phi = \psi \wedge \theta$}{
			$Y\gets X$\;
			\While{$Y \neq \KwFn{$\KwFn{$Y,\psi$},\theta$}$}{
				$Y \gets \KwFn{$\KwFn{$Y,\psi$},\theta$}$\;
			}
			\Return $Y$\;
		}
		\uElseIf{$\phi = p$}{
			\Return  $\{s\in X\mid s(p)=1\}$\;
			}
		\uElseIf{$\phi = \neg p$}{
			\Return  $\{s\in X\mid s(p)=0\}$\;
		}
		\ElseIf{$\tuple p \sub \tuple q$}{
			$Y\gets X$\;
			\While{$Y\neq \{s\in Y \mid \exists s'\in Y: s(\tuple p)=s'(\tuple q)\}$}{
				$Y\gets \{s\in Y \mid \exists s'\in Y: s(\tuple p)=s'(\tuple q)\}$\;
			}
			\Return $Y$
		}
  			
\end{algorithm}

\begin{lem}\label{lem:upper_qplinc}
The entailment problem for $\QPLInc$ is in $\coNEXPTIME$.
\end{lem}
\begin{proof}
Consider the computation of $\texttt{MaxSub}(X,\phi)$ in Algorithm \ref{alg:qplinc_entail}. We leave it to the reader to show, by straightforward induction on the complexity of $\phi$, that for all $X,Y$ over a shared domain $V\supseteq \Fr{\phi}$ the following two claims holds: 
\begin{enumerate}
\item $\texttt{MaxSub}(X,\phi)\models \phi$, and 
\item $Y\sub X $ and $ Y\models \phi \Rightarrow Y\sub \texttt{MaxSub}(X,\phi)$. 
\end{enumerate}
Note that $\texttt{MaxSub}(X,\phi)$ is the unique maximal subteam of $X$ satisfying $\phi$.  The idea is that each subset of $X$ satisfying $\phi$ survives each iteration step. Since we have $\texttt{MaxSub}(X,\phi)\sub X$, it now follows directly from (1) and (2) that $\texttt{MaxSub}(X,\phi)= X$ iff $X\models \phi$. 

Let us now present the universal exponential-time algorithm for deciding entailment for $\QPL$. Assuming an input sequence $\phi_1, \ldots ,\phi_n$ from $\QPLInc$, the question is to decide whether $\{\phi_1, \ldots ,\phi_{n-1}\}\models \phi_n$. The algorithm first universally guesses  a team $X$ over $\bigcup_{i=1}^n \Fr{\phi_i}$ and then  using  Algorithm \ref{alg:qplinc_entail} deterministically  tests whether $X$ is a counterexample for $\{\phi_1, \ldots ,\phi_{n-1}\}\models \phi_n$, returning true iff this is not the case. Note that $X$ is a counterexample for $\{\phi_1, \ldots ,\phi_{n-1}\}\models \phi_n$ iff $\texttt{MaxSub}(X,\phi_n)\neq X$ and $\texttt{MaxSub}(X,\phi_i)= X$ for $i=1, \ldots ,n-1$. By the locality principle of $\QPL$ (Proposition \ref{prop:qpl_locality}) this suffices, \ie, each universal branch returns true iff $\{\phi_1, \ldots ,\phi_{n-1}\}\models \phi_n$. 

It remains to show that the procedure runs in exponential time. Consider first the running time of Algorithm \ref{alg:qplinc_entail} over an input $(X,\phi)$. First note that  one can find an exponential function $g$ such that at each recursive step $\texttt{MaxSub}(Y,\psi)$ the size of the team $Y$ is bounded by $|X|g(|\phi |)$. The possible exponential blow-up comes from nested quantification in $\phi$. Furthermore, each base step $\texttt{MaxSub}(Y,\psi)$ can be computed in polynomial time in the size of $Y$. Also, each recursive step $\texttt{MaxSub}(Y,\psi)$ involves at most $|Y|$ iterations consisting of either computations of $\texttt{MaxSub}(Z,\theta)$ for $Z\sub Y$ and a subformula $\theta$ of $\psi$, or removals of assignments from $Y$. It follows by induction that there exists a polynomial $f$ and an exponential $h$ such that the running time of $\texttt{MaxSub}(Y,\psi)$ is bounded by $f(|X|)h(|\phi|)$.  The overall algorithm now guesses first a team $X$ whose size is possibly exponential in the input. By the previous reasoning, the running time of $\texttt{MaxSub}(X,\phi_i)$ remains exponential for each $\phi_i$. This shows the claim.
 % and (iii) $Y\sub X \Rightarrow  \texttt{MaxSub}(Y,\phi) \sub \texttt{MaxSub}(X,\phi)$. The cases for $\phi = p$ and $\phi = \neg p$ are clear. The remaining are considered below:
%\begin{itemize}
%\item $\phi=\tuple p \sub\tuple q$: It is obvious that the fixed point $\texttt{MaxSub}(X,\phi)$ satisfies $\phi$. Also, for all $Y\sub X$ we have that $\{s\in Y \mid s(\tuple p)\sub Y(\tuple q)\}\sub \{s\in X \mid s(\tuple p)\sub X(\tuple q)\}$; if in addition $Y\models \phi$, then we have that  $Y\sub \{s\in X \mid s(\tuple p)\sub X(\tuple q)\}$. Hence (ii,iii) follows by simple induction.
%\item $\phi =\psi \wedge\theta$: By the induction assumption we obtain that the fixed point $\texttt{MaxSub}(X,\phi)$ satisfies $\phi$; also (ii,iii) are straightforward by the induction assumption.
%\item $\phi =\psi \vee \theta$: By the induction assumption we immediately obtain that $\texttt{MaxSub}(X,\phi)$ satisfies $\phi$; again, (ii,iii) are straightforward by the induction assumption.
%\item $\phi =\exists p \psi$: By the induction assumption we immediately obtain that $\texttt{MaxSub}(X,\phi)$ satisfies $\phi$; again, (ii,iii) is straightforward by the induction assumption.
%\item $\phi =\forall p \psi$: By the induction assumption we immediately obtain that $\texttt{MaxSub}(X,\phi)$ satisfies $\phi$; again, (ii,iii) is straightforward by the induction assumption.
%\end{itemize}
\end{proof}
Lemmata \ref{lem:upper_qplinc} and \ref{lem:lower_qplinc} now show the $\coNEXPTIME$-completeness of the $\QPLInc$ entailment problem. The same lower bound have been shown in \cite{HKMV16} to apply already to $\QPLInc$ validity. The exact complexity of $\MLInc$ validity and entailment however remains  an open problem.

%Theorem \ref{sdwerwerfsdf}
%has been proven in \cite{HKMV16}, and for entailment the result also follows by Theorems \ref{ssdfddf} and \ref{thm:reduction}. In contrast to the lower bound for validity, notice that satisfiability for $\MLInc$ is  $\EXPTIME$-complete \cite{HellaKMV15}.

\begin{restatable}{thm}{akuus}
\label{ssdfddf}
The entailment problem for $\QPLInc$ is   $\coNEXPTIME$-complete. 
\end{restatable}
\begin{thmC}[\cite{HKMV16}]
\label{sdwerwerfsdf}
The entailment and the validity problems for $\MLInc$ are $\coNEXPTIME$-hard.
\end{thmC}

\section{Conclusion}\label{sect:conclusion}

%
%\begin{table*}[t]
%\begin{center}
%\scalebox{.8}{
%\begin{tabular}{llll}
%\toprule
%	 &  satisfiability& validity & entailment  \\\midrule
%	 $\PL$ & $\NP$ \cite{Cook71,Levin73} & $\coNP$ \cite{Cook71,Levin73} & $\coNP$ \cite{Cook71,Levin73} \\
%	  	$\ML$ &$\PSPACE$ \cite{Ladner77} & $\PSPACE$ \cite{Ladner77} & $\PSPACE$ \cite{Ladner77} \\
%	  $\MLDis$ & $\PSPACE $ \cite{sevenster09b} & $\PSPACE$ [Thm. \ref{thm:mldis}]  & $\PSPACE$ [Thm. \ref{thm:mldis}] \\
%	$\PLDep$  & $\NP$ \cite{LohmannV13} &$\NEXPTIME$ \cite{Virtema14} & $\coNEXPTIME^{\NP}$ [Thm. \ref{thm:entail}] \\
 %	$\QPLDep,\MLDep,\EMLDep$ & $\NEXPTIME$ \cite{Peterson2001,sevenster09b}, [Thm. \ref{thm:reduction}] & $\NEXPTIME$ [Thm. \ref{cor:val_mldep}]&  $\coNEXPTIME^{\NP}$ [Thm. \ref{thm:entail}]  \\
% \bottomrule
%\end{tabular}}
%\vspace{0ex}
%\caption{Summary of results. The stated complexity classes refer to completeness results.}\label{newresults}
%\end{center}
%\vspace*{-1ex}%
%\end{table*}

\begin{table*}[t]
\begin{center}
\scalebox{.7}{
\begin{tabular}{@{}llll@{}}

\toprule
	 &  satisfiability& validity & entailment  \\\midrule
	 $\PL$ & $\NP$ \cite{Cook71,Levin73} & $\coNP$ \cite{Cook71,Levin73} & $\coNP$ \cite{Cook71,Levin73} \\
	  	$\ML$ &$\PSPACE$ \cite{Ladner77} & $\PSPACE$ \cite{Ladner77} & $\PSPACE$ \cite{Ladner77} \\
	  $\MLDis$ & $\PSPACE $ \cite{sevenster09b} & $\PSPACE$ [Thm. \ref{thm:mldis}]  & $\PSPACE$ [Thm. \ref{thm:mldis}] \\
	$\PLDep$  & $\NP$ \cite{LohmannV13} &$\NEXPTIME$ \cite{Virtema14} & $\coNEXPTIME^{\NP}$ [Thm. \ref{thm:entail}] \\
 	$\QPLDep,\MLDep,\EMLDep$ & $\NEXPTIME$ \cite{Peterson2001,sevenster09b}, [Thm. \ref{thm:reduction}] & $\NEXPTIME$ [Thm. \ref{cor:val_mldep}]&  $\coNEXPTIME^{\NP}$ [Thm. \ref{thm:entail}]  \\
 	$\QPLInd$ &$\NEXPTIME$ \cite{Peterson2001,KontinenMSV14}, [Thm. \ref{thm:reduction}] &  $\coNEXPTIME^{\NP}$ [Thm. \ref{sdfsdf}]&  $\coNEXPTIME^{\NP}$ [Thm. \ref{sdfsdf}]\\
 	$\MLInd$ & $\NEXPTIME$ 
 	\cite{KontinenMSV14} &$\geq \coNEXPTIME^{\NP}$ [Cor. \ref{sdkf}]& $\geq \coNEXPTIME^{\NP}$ [Cor. \ref{sdkf}]\\
 	$\QPLInc$ &$\EXPTIME$ \cite{HellaKMV15}, [Thm. \ref{thm:reduction}]& $\EXPTIME$ [Thm. \ref{thm:valqplinc}] & $\coNEXPTIME$ [Thm. \ref{ssdfddf}]  \\
 	$\MLInc$ & $\EXPTIME$ \cite{HKMV16}&$\geq \coNEXPTIME$ \cite{HKMV16}&$\geq \coNEXPTIME$ \cite{HKMV16} \\\bottomrule

\end{tabular}
}
\vspace{2ex}
\caption{Summary of results. The stated complexity classes refer to completeness results, except that the prefix \enquote{$\geq$} refers to hardness results.}\label{newresults}
\end{center}
%\vspace*{-5ex}%
\end{table*}
We have examined the validity and entailment problem for various modal and propositional dependence logics (see Table \ref{newresults}). We showed that the entailment problem for (extended) modal and (quantified) propositional dependence logic is $\coNEXPTIME^{\NP}$-complete, and that the corresponding validity problems are  $\NEXPTIME$-complete. We also showed that modal logic extended with intuitionistic disjunction is $\PSPACE$-complete with respect to its satisfiability, validity, and entailment problems, therefore being not more complex than the standard modal logic. Furthermore, we examined extensions of propositional and modal logics with independence and inclusion atoms. 
Quantified propositional independence logic was proven to be $\coNEXPTIME^{\NP}$-complete both in terms of its validity and entailment problem. For quantified propositional inclusion logic the validity and entailment problems were shown to be $\EXPTIME$-complete and $\coNEXPTIME$-complete, respectively. Using standard reduction methods we established the same lower bounds for modal independence and inclusion logic, although for validity of modal inclusion logic a higher lower bound of $\coNEXPTIME$ is known to apply.
%The same lower bounds apply to the corresponding modal logics.  %We also gave sound and complete axiomatizations for the entailment and validity problems of extended modal and (quantified) propositional dependence logic. These axiomatizations are optimal in the sense that they reflect the  $\coNEXPTIME^{\NP}$ entailment and $\NEXPTIME$ validity algorithms. 
However, we leave determining the exact complexities of validity/entailment of $\MLInd$ and $\MLInc$ as an open problem.  It is plausible that solving these questions will  open up possibilities for novel axiomatic characterizations.

%%
%% Bibliography
%%

%% Either use bibtex (recommended), 

\bibliographystyle{plainurl}% the recommended bibstyle
\bibliography{biblio}

%\input{content/appendix}

%% .. or use the thebibliography environment explicitely

\end{document}